\documentclass[12pt]{article}

\usepackage[usenames,dvipsnames]{color}
\usepackage[round]{natbib}
\usepackage[hyperfootnotes=false]{hyperref}
\usepackage{verbatim}    
\usepackage{caption}
\usepackage{threeparttable}
\usepackage{xcolor}
\usepackage{subcaption}
\usepackage{rotating}
\usepackage{hyperref}    
\usepackage{enumerate}   
\usepackage{float}
\usepackage[final]{pdfpages}
\usepackage{ragged2e}
\usepackage{lscape}

\newcommand{\tabincell}[2]{\begin{tabular}{@{}#1@{}}#2\end{tabular}}

\usepackage{pifont}
\newcommand{\xmark}{\ding{55}}%

\usepackage{fullpage}
 \usepackage[top=0.9in, bottom=0.9in, left=0.9in, right=0.9in]{geometry}

 \usepackage{setspace} 
\usepackage{textcomp}
\usepackage{amsbsy}      
\usepackage{amssymb}
\usepackage{amsmath}
\usepackage{amsthm}      

\usepackage[all]{xy}     
\usepackage{color}       
\usepackage{graphicx}    

\usepackage{fancyhdr}         
\usepackage{lastpage}    
\usepackage{natbib}
\usepackage{multirow}
\usepackage{mathrsfs}
\usepackage{epsfig}
\usepackage{bbm}
\hypersetup{colorlinks,citecolor=Blue,linkcolor=Blue,urlcolor=Blue}

\usepackage{lscape}

\theoremstyle{plain} 
\newtheorem{thm}{Theorem}
\newtheorem{asp}{Assumption}

\newtheorem{prop}{Proposition}
\newtheorem{coro}{Corollary}

\begin{document}

\title{\large{\textbf{
Eliciting Information from Sensitive Survey Questions}%
}}
\author{Yonghong An\thanks{Department of Economics, Texas A\&M University. E-mail address: \href{yonghongan@tamu.edu}{yonghongan@tamu.edu}.} \and Pengfei Liu\thanks{Department of Environmental and Natural Resource Economics, University of Rhode Island. Email: \href{mailto:pengfei\_liu@uri.edu}{pengfei\_liu@uri.edu}.}}
\date{\today}
\maketitle


\begin{singlespace}
\begin{abstract}
This paper considers how to elicit information from sensitive survey questions. First we thoroughly evaluate list experiments (LE), a leading method in the experimental literature on sensitive questions. Our empirical results demonstrate that the assumptions required to identify sensitive information in LE are violated for the majority of surveys. Next we propose
a novel survey method, called Multiple Response Technique (MRT), for eliciting information from sensitive questions. We require all of the respondents to answer three questions related to the sensitive information. This technique recovers sensitive information at a disaggregated level while still allowing arbitrary misreporting in survey responses. An application of the MRT provides novel empirical evidence on sexual orientation and Lesbian, Gay, Bisexual, and Transgender (LGBT)-related sentiment.
\end{abstract}


\end{singlespace}

\break

\section{Introduction}
Surveys designed to address a wide range of social and economic issues---such as racial prejudice, drug use, illegal immigration, and fraud---often elicit information on private, illegal, or socially undesirable behaviors. However, respondents tend to misreport their views on sensitive topics, and that misreported information often leads to biased estimates of the behaviors, their determinants, and their causal effects.

Estimating sensitive information from potentially misreported data has important implications for social and economic policy decisions. Therefore, methods that produce unbiased estimates are prominent. The method of list experiments (LE), also known as the ``item count technique"  and ``unmatched count
technique", is one of the most popular methods used to reduce response bias in surveys. It was proposed in \cite{raghavarao1979block} to enhance respondents' willingness to answer truthfully by anonymizing the survey. In LE, a sample of respondents are randomly assigned to a control and a treatment group. The control group receives
a set of nonsensitive questions; the treatment group receives the same set of questions plus one sensitive question. The respondents in both groups only indicate the number of applicable items.
 \cite{coffman2017size} further develop LE by modifying the survey design: the control group is required to answer
the sensitive question directly, in addition to indicating the number of applicable nonsensitive items.

LE has been widely used in the past three decades in multiple disciplines, including political science, sociology, psychology, and statistics, and it has been applied increasingly in economics recently.\footnote{\cite{blair2018worry} summarize studies using LE in multiple disciplines.} When searched the keywords ``list experiments", ``item count technique", and ``unmatched count technique", we found 20 economics articles using the method after 2010, with 13 of them published after 2016. \textit{The American Economic Review} and \textit{Quarterly Journal of Economics}
had no publications using LE before 2016, but published six articles using it during 2016-2019.

Despite its wide applications, LE has never been assessed systematically for its validity. Thus the reliability of empirical results
based on the LE method is unclear. In this paper, we develop a method to evaluate LE and then show that the assumptions required to identify sensitive information in LE are violated for the majority of its empirical applications. In addition to a high likelihood of failure, LE, and other existing methods of eliciting sensitive information only provide a measure of sensitive information at the aggregate level. This limits the use of the aggregated measure as an outcome variable (\cite{bertrand2017field}). To address these issues, we propose a novel method, multiple response technique (MRT), to elicit sensitive information from survey questions. Our method is applicable to surveys with measurement error due to misreporting, regardless of the presence of sensitive information.

We begin by providing a rigorous framework for LE. First, we generalize LE by relaxing a ubiquitous assumption that respondents will answer survey questions truthfully. In our generalized model, respondents in the control and treatment groups may misreport, and may do so in different ways. Truthful response is nested as a special case of the generalized model where the misreporting probability equals zero. We show that in the generalized model, if respondents are randomly assigned into two groups and the sensitive item does not alter the preferences of respondents in the treatment group, then restrictions on the observed responses imposed by LE can be summarized as a set of moment conditions. The model's parameters, including the proportion of respondents with sensitive information and misreporting probabilities, can be estimated from the moment conditions through the generalized method of moments (GMM) if there are at least three nonsensitive questions.  We find that using mean difference (difference between means of applicable items in the two groups), the commonly used approach to estimating the proportion of respondents with sensitive information is generally biased in the presence of misreporting. The magnitude of the bias depends on the true parameters of the model, the number of nonsensitive questions, and misreporting probabilities.

Based on the moment conditions just described, we develop a procedure for testing the validity of the basic assumptions of LE.
 We apply our test to five recently published articles where the results depend partially on LE. For the majority of the empirical applications, the fundamental assumptions are violated. Thus, estimates based on LE through different methods--including mean difference, ordinary least squares (OLS), nonlinear least squares (NLS), and maximum likelihood estimation (MLE)--would be biased and the implications of those estimates, together with their determinants and causal effects, are problematic.
The problems with LE
could be due to
the significance of misreporting, the non-randomness of group assignment, or the impacts of the sensitive item on respondents' preference. However, in the literature it is unclear how we can make adjustments to LE to obtain unbiased estimates once its
fundamental assumptions fail.

Therefore, we propose a novel method of eliciting information from sensitive survey questions.\footnote{It is worth noting that our method is also applicable in those settings where LE is not rejected. However, one should be cautious in comparing the results of the two methods because they are based on
different sets of assumptions.}
Respondents receive three or more survey questions related to the sensitive information. Random assignment to groups is not required. We
treat the sensitive information as respondents' \textit{unobserved heterogeneity} and recover it from their responses to multiple survey questions. Our method's validity is based on two assumptions: (1) after we control the set of observed (e.g., demographic) variables and the sensitive information, the respondent's responses to three survey questions will be independent; and (2) respondents with or without sensitive information will answer some survey questions differently. Given a sample of the survey responses and respondents' characteristics, we then propose two versions of estimation suitable for the different data patterns. When the covariates are discrete, it is convenient to estimate model primitives using a closed-form and nonparametric procedure. One prominent property of the nonparametric approach is that it is global and involves no numerical optimizations. When the covariates are continuous, we apply maximum likelihood estimation (MLE) to estimate the dependence of sensitive behaviors on the covariates. The resulting estimates allow us to predict the likelihood of sensitive behaviors for a given value of the covariates. Monte Carlo simulations show good performance of both estimating procedures with a modest sample size.

Our proposed method has several important advantages over LE and other approaches to eliciting sensitive information. First, we can recover sensitive information at a disaggregated level. For example, in a study of respondents' attitude toward
same-sex marriage, we are able to recover the proportions of supporters for both LGBT and non-LGBT populations, while the existing methods only provide an overall proportion. Second, our proposed method allows misreporting in an arbitrary form, and the measurement error due to that misreporting can be recovered. When a sample of respondents are asked about their sexuality in a survey question, we can estimate the measurement error of their responses.
And last but not least, the effects of group assignment on survey results are ruled out because respondents are not grouped.

Using our proposed method we estimate sexual orientation and LGBT-related sentiment.
The data were analyzed in \cite{coffman2017size}, where the main objective was to illustrate the substantially underestimated
size of the LGBT-population and the magnitude of anti-gay sentiment.
We apply our proposed technique to this data and our analysis leads to several novel findings that are  obscured with LE and other existing survey methods.

Ours are the first set of quantitative results in the literature on respondents' misreporting behavior when indicating their sexual orientation, and on how the dependence of misreporting behavior depends on demographics. A
substantial portion (about 28.8\%) of non-heterosexual respondents report themselves to be heterosexual, while heterosexual respondents report truthfully.  Male, Black, Christian, and Republican respondents report a much lower proportion of non-heterosexuality than their counterpart groups. For example,
47.3\% of non-heterosexual Republicans claim themselves to be heterosexual while the comparable percentage is only 25.6\% for Democrats.
These results are obtained without any \textit{ex ante} information on respondents' sexual orientation.

Respondents with negative sentiments to the LGBT population are significantly divided. Black, Christian, and Republican respondents are substantially more negative toward the LGBT population than their counterparts: for example, 22\% of Christian respondents have negative sentiments, while the proportion is only 3.2\% for non-religious respondents. We also find significant divergence on
how respondents with negative sentiments toward the LGBT population respond to the three sentiment-related questions. ``Happy with LGB manager"
is accepted by only 2\% of the respondents, while 17.8\% of them support same-sex marriage and 58.5\% of them ``believe it is illegal to discriminate against LGBT people".   These estimates exhibit sharp divides among demographic groups, and the attitude toward same-sex marriage is the most diverse for the three questions: the proportion of white respondents who support same-sex marriage is five times that of Black respondents; for nonreligious respondents it is three times that of Christian; and for Democrats is five times that of Republicans.

This paper contributes to a broader literature on surveys and related empirical studies. First of all, we develop a rigorous procedure
 to check whether LE is applicable to the survey data. There are now two strands of literature on the applications of LE. The first focuses on estimating the proportion of sensitive behavior, its determinants and consequences. The applications cover a wide range of topics in economics, including anti-authoritarian
movements (\cite{cantoni2019protests}); criminal behavior (\cite{kuha2014item}); the effects of labor exclusion on responsibilities (\cite{ronconi2015labor}); election (\cite{neggers2018enfranchising} and \cite{barrera2020facts}); evaluation of anti-poverty programs (\cite{haushofer2016short} and \cite{muralidharan2016building});
impacts of media censorship (\cite{chen2019impact});  influence of exporting
democracy (\cite{humphreys2019exporting}); microfinance (\cite{karlan2012list} and \cite{karlan2016follow}); relationship between financial conditions and performance of firms (\cite{bilir2019host}); and sexual health (\cite{chong2013effectiveness} and \cite{treibich2019estimating}).
The second strand relies on the modified LE in \cite{coffman2017size} to justify the reliability of survey data. If the proportion of affirmative answers to the sensitive question in the control group is the same as the mean difference, then misreporting does not exist or is negligible, e.g., see
\cite{cantoni2019protests} and \cite{chen2019impact} for applications. All of these applications assume that the basic assumptions of LE hold, without checking their validity.  To the best of our knowledge, we are the first to empirically evaluate the validity of the assumptions in the LE approach. Our negative findings suggest that researchers should be cautious when applying LE and interpreting empirical results based on LE.

Second, although our proposed technique is motivated by the potential failure of LE, the applicability of our MRT technique is more general. Our technique has its flexibility to address various types of measurement errors due to misreporting.
Most of the empirical studies on experimental data are silent on measurement error despite its universality
in survey data (\cite{bound2001measurement}).
Notable exceptions are
\cite{blattman2016measuring} and  \cite{blattman2017reducing}: they develop a validation technique
 to estimate measurement error and apply it to study the impacts of
behavioral therapy on crime and violence. Nevertheless, their validation approach is based on the availability
of a subsample without measurement error, and that is obtained through in-depth participation observations. This type of validation method is powerful when analyzing data with measurement error (e.g., see similar studies in \cite{bollinger1998measurement}, \cite{bound2001measurement}, and \cite{chen2005measurement}), but its applicability is limited because a validation sample without measurement error is rare. By contrast, our technique allows measurement error in flexbile forms and thus opens a new avenue for analysis of experiment data with misreporting.

Our paper also contributes to the literature on measurement error. Our technique is based on a recently developed methodology for measurement error in \cite{hu2008identification}. The method has been used widely to recover unobserved heterogeneity and unobservables in industrial organization and labor economics: (\cite{hu2017econometrics} provides an excellent review of applications of this method). For example, \cite{feng2013misclassification} apply this method to correct
U.S. unemployment rates by addressing misclassification in self-reported labor force status.  Nevertheless, we are the first to
treat sensitive information as the unobserved heterogeneity of respondents and to recover it using the methodology. Our technique serves as a first attempt to connect the literature on measurement error models and experimental data based on survey methods, and sheds light on the importance of communication between the two branches  of the literature.

The remainder of the paper is organized as follows.
In Section 2, we analyze LE and evaluate the validity of LE assumptions in existing empirical studies. In Section 3,
we present our new technique and demonstrate its performance. Section 4 reports our empirical findings
on sexuality and LGBT-related sentiments. Section 5 concludes. Proofs, simulation results, and additional discussions of LE are in the Appendix.

\section{List Experiments\label{section: LE}}

\subsection{The method and properties\label{subsection: setup and properties}}
There is a population of $n$ respondents with characteristics summarized by a vector $Z$. The respondents are randomly assigned to a control group and a treatment group. There are $J$
 non-sensitive yes-no questions for the control group and $J+1$ yes-no questions for the treatment group: the same $J$ non-sensitive questions plus one sensitive question.  Let $t\in\{0,1 \}$ be a binary indicator of the randomized treatment assignment; $t=0$ and $t=1$ indicate that a respondent is assigned to the control group or the treatment group.
We use $Y_t\in \mathcal{J}_t\equiv \{0, 1, \cdots, J, J+t\}$ to denote the observed outcome (yes-no response) in group $t$ and $P_t(j)\equiv \Pr(Y_t=j)$ is the probability of outcome $j$ being observed in group $t$.

Let $R_0 \in\mathcal{J}_0$ and $R_1 \in\mathcal{J}_0$ be the random variables that describe respondents' responses to the non-sensitive questions in the control and treatment group, respectively. $X^*\in\{0,1\}$ is a random variable that reflects respondents' true preference regarding the sensitive question. Let $\Pr(X^*=1)\equiv \delta$ be the proportion of respondents who answer the sensitive question affirmatively under their true preference. A ubiquitous assumption in the literature is that adding the sensitive question does not change responses to the nonsensitive questions in the treatment group, formalized by the assumption below.

\begin{asp}\label{asp: random assignment}
Under the true preference, the presence of a sensitive question does not change the distribution of respondents' responses to the nonsensitive questions, i.e, $\Pr(R_0=j)=\Pr(R_1=j)$, $\forall j\in\mathcal{J}_0$.
\end{asp}
There are two restrictions imposed by Assumption \ref{asp: random assignment}.
First, assignment to the control and treatment groups is random, i.e., the respondents in the two groups have the same preference regarding the nonsensitive questions. Second, the addition of a sensitive question has no impact on respondents' true preference as to the non-sensitive questions. Assumption \ref{asp: random assignment} can be relaxed to a weaker version: the true preferences of respondents' to the nonsensitive questions are the same conditional on the vector of characteristics $Z$, i.e., $\Pr(R_0=j|Z=z)=\Pr(R_1=j|Z=z)$. Our analysis of LE can be extended to the case where Assumption \ref{asp: random assignment} holds conditional on $Z$.
Various versions of Assumption \ref{asp: random assignment} are imposed but not explicitly stated in most LE studies. One exception is \cite{blair2012statistical}, where a stronger version of the assumption is used: every respondent's response to the nonsensitive questions would be the same in the two groups.

When respondents reveal their preference truthfully in answering the survey questions, we have \begin{eqnarray}\label{equation: true preference}\Pr(R_0=j)=P_0(j), j\in\mathcal{J}_0; \hspace{0.2cm}\Pr(R_1+X^*=j)=P_1(j), j\in\mathcal{J}_1.\end{eqnarray} However, misreporting is common in self-reported surveys \citep{bound2001measurement}, even when questions are non-sensitive. We show that equation (\ref{equation: true preference}) no longer holds under misreporting. Let $p_t\in[0,1)$ be the probability of misreporting in the group $t$.
\begin{asp}\label{asp: misreporting probability}
Respondents in group $t\in\{0,1\}$ truthfully respond with probability $1-p_t$ and misreport with probability $p_t$. Respondents choose each possible response with equal probability when they misreport.
\end{asp}
This assumption is motivated by empirical findings on validation studies in the literature of measurement error (e.g., \cite{bollinger1998measurement}, \cite{bound2001measurement}, and \cite{chen2005measurement}) that survey respondents both report truthfully and misreport, intentionally or unintentionally, with positive probabilities. In Assumption \ref{asp: misreporting probability} we assume that respondents misreport unintentionally because of their inattention or lack of effort. Thus, a natural specification is that each possible outcome is chosen with equal probability. Note that we still allow the sensitive question to affect respondents' misreporting behavior; i.e.,
the misreporting probabilities differ in the two groups, $p_0\neq p_1$.\footnote{\cite{blair2019list} introduce a ``uniform error" to LE models by assuming $p_0=p_1$.}

Our analysis of LE can be readily extended to a model with other types of misreporting errors given that misreporting strategies are known. In Appendix \ref{appendix: an alternative misreporting strategy}, we discuss one case of intentional or strategic misreporting in which respondents misreport only if their truthful responses disclose privacy, i.e.,  when their truthful response is the outcome $J+1$.

Under Assumptions \ref{asp: random assignment} and \ref{asp: misreporting probability}, we establish in the following proposition the connection between the distributions of responses in the two groups, allowing for misreporting behaviors.
\begin{prop}\label{proposition: model restrictions under randomization}
Under Assumptions \ref{asp: random assignment} and \ref{asp: misreporting probability}, the observed distribution of  responses in the two groups satisfies the following restrictions,
\begin{eqnarray}\label{equation: relationship between observed responses with misreporting}
P_1(0)&=&\frac{1-p_1}{1-p_0}\bigg((1-\delta)P_0(0)-\frac{(1-\delta)p_0}{J+1}\bigg)+\frac{p_1}{J+2},\nonumber\\
P_1(j)&=&\frac{1-p_1}{1-p_0}\left(\delta P_0(j-1)+(1-\delta)P_0(j)-\frac{p_0}{J+1}\right)+\frac{p_1}{J+2},\hspace{0.1cm}j=1,2,\dots,J,\nonumber\\
P_1(J+1)&=&\frac{1-p_1}{1-p_0}\bigg(\delta P_0(J)-\frac{\delta p_0}{J+1}\bigg)+\frac{p_1}{J+2},
\end{eqnarray}
where $0\leq P_t(j)\leq 1$,\hspace{0.2cm} $\sum_{j=0}^{J+t} P_t(j)=1$,  $t=0, 1$.
\end{prop}
\begin{proof}
See Appendix \ref{appendix_section:proof}.
\end{proof}
Equation (\ref{equation: relationship between observed responses with misreporting}) summarizes all of the restrictions that LE imposes on the observed responses under Assumptions \ref{asp: random assignment} and \ref{asp: misreporting probability}. Equation (\ref{equation: relationship between observed responses with misreporting}) also nests the special cases where the probabilities of misreporting in the two groups are the same ($p_0=p_1$) and there is no misreporting ($p_0=p_1=0$). Proposition \ref{proposition: model restrictions under randomization} can be extended to alternative, non-random misreporting strategies. Equation (\ref{equation: relationship between observed responses with misreporting}) will be modified accordingly, based on the alternative misreporting strategy, to reflect the connection between the distribution of responses in the two groups.

One widely used approach to estimating the probability of sensitive information is to compute the
 difference between the mean of responses in the two groups.\footnote{\cite{imai2011multivariate} and \cite{imai2015using} propose to estimate LE with multiple covariates by NLS and MLE, respectively. Both approaches still rely on Assumptions \ref{asp: random assignment} and \ref{asp: misreporting probability}.} The corollary below shows that the mean difference approach may lead to a biased estimate in the presence of misreporting.\footnote{Without using the relationship in equation (\ref{equation: relationship between observed responses with misreporting}), \cite{blair2019list} show that when $p_0=p_1\neq 0$, the mean difference is a biased estimator of $\delta$.}
\begin{coro} \label{corollary: mean difference with misreporting}
Suppose Assumptions \ref{asp: random assignment} and \ref{asp: misreporting probability} hold, the mean difference of responses in the two groups
is
\begin{eqnarray}\label{equation: mean difference of responses with misreporting}
\mathbb{E}(Y_1)-\mathbb{E}(Y_0)&\equiv&\sum\nolimits_{j=0}^{J+1}P_1(j)j-\sum\nolimits_{j=0}^JP_0(j)j\nonumber\\
&=&\delta-p_1\delta-\frac{J(1-p_1)p_0}{2(1-p_0)}-\frac{p_1-p_0}{1-p_0}\mathbb{E}(Y_0)+\frac{J+1}{2}p_1.
\end{eqnarray}
When $p_0=p_1=p\neq 0$, the mean difference is
$\mathbb{E}(Y_1)-\mathbb{E}(Y_0)=\delta+p(1-2\delta)/2$.
\end{coro}
\begin{proof}
See Appendix \ref{appendix_section:proof}.
\end{proof}
This corollary states that the mean difference is determined by misreporting probabilities $p_0$ and $p_1$, the mean of responses in the control group $\mathbb{E}(Y_0)$, and the number of nonsensitive questions $J$.
The mean difference does not identify the parameter $\delta$ unless (1) there is no misreporting, i.e., $p_0=p_1=0$ or (2) there is no misreporting in the treatment group ($p_1=0$) and the mean response in the control group is $J/2$, i.e., $\mathbb{E}(Y_0)=J/2$. A sufficient condition of $\mathbb{E}(Y_0)=J/2$ is that respondents choose each outcome from $\{0, 1, \cdots, J\}$ with equal probability. \cite{coffman2017size} show that, in most current surveys, misreporting on stigmatized opinions is inevitable. Moreover, when the probabilities of misreporting $p_0$ and $p_1$ are unknown, the approach to estimating $\delta$ by the mean difference is not reliable. The bias of a mean difference estimate could be in either direction and is unknown \textit{ex ante}.

The result in Corollary \ref{corollary: mean difference with misreporting} reconciles the results of
LE's dependence on the number of non-sensitive questions. \cite{gosen2014social} shows that the results of LE depend on the number of non-sensitive questions, while \cite{tsuchiya2007study} find that LE does not rely on the number of non-sensitive questions.  We show that the mean difference generally does rely on the number of nonsensitive questions unless respondents reveal their preferences truthfully.

Equation (\ref{equation: relationship between observed responses with misreporting}) provides $J+2$ conditions for three unknown parameters $(\delta, p_0, p_1)$ with one of the conditions being redundant. As a result, when there are at least three non-sensitive questions ($J\geq 3$), $\delta, p_0$, and  $p_1$ are identified from equation (\ref{equation: relationship between observed responses with misreporting}) and GMM can be applied to estimate the parameters under misreporting. We present the proof of identification in Appendix \ref{appendix_section:proof}.

\subsection{A Test of List Experiments\label{subsection: test misreporting behavior}}
Assumptions \ref{asp: random assignment} and \ref{asp: misreporting probability} are sufficient conditions for
equation (\ref{equation: relationship between observed responses with misreporting}), which summarizes the model restrictions and serves as a basis for identification of the model parameters $\delta$, $p_0$, and $p_1$ through various estimation methods: OLS, NLS, MLE, or GMM.
Before estimation, we need to test whether the data satisfies the restrictions in
equation (\ref{equation: relationship between observed responses with misreporting}), i.e., whether there is a unique solution to equation (\ref{equation: relationship between observed responses with misreporting}) for the given dataset. A rejection of equation (\ref{equation: relationship between observed responses with misreporting}) implies that there does not exist a triplet $(\delta, p_0, p_1)$ that satisfies the equation and that at least one of the two assumptions is violated. In this case, LE cannot be used to estimate sensitive information. Below we propose a procedure to test whether the data is consistent with the restrictions imposed by equation (\ref{equation: relationship between observed responses with misreporting}). We focus on the model with more than two non-sensitive questions ($J\geq 3$) since parameters cannot be identified from equation (\ref{equation: relationship between observed responses with misreporting}) when $1\leq J\leq 2$  .

Let $D_{i}=(Y_{i}, Z_{i}, t_i), i=1, 2, \cdots, n$ be an i.i.d. sample of $(Y, Z, t)$, where $Y$, $Z$, and $t$ are respondents' responses, characteristics, and a group indicator, respectively. Let $\theta\equiv(\delta, p_0, p_1)\in [0,1]^3\equiv \Omega$, Proposition \ref{proposition: model restrictions under randomization} states that any $\theta\in  \Omega$ that satisfies equation (\ref{equation: relationship between observed responses with misreporting}) is consistent with the data $\{D_i\}_{i=1}^n$. Let $\Theta$ be the identified set of the parameter that rationalizes the data.  When $J\geq 3$, equation (\ref{equation: relationship between observed responses with misreporting}) has at most one solution: the intersection of $\Theta$ and $\Omega$ is a singleton or empty set. If $\Theta\cap\Omega\neq\emptyset$, we may estimate $\theta$ from equation (\ref{equation: relationship between observed responses with misreporting}) by GMM. If $\Theta\cap\Omega=\emptyset$, any estimation method based on equation (\ref{equation: relationship between observed responses with misreporting}) is invalid. We propose a $J$-test to conduct inference of the existence of a solution to equation (\ref{equation: relationship between observed responses with misreporting}).
\begin{eqnarray}
H_0:  \Theta\cap\Omega\neq \emptyset\hspace{0.5cm} \hbox{vs.} \hspace{0.5cm}
 H_1:  \Theta\cap\Omega=\emptyset.
\end{eqnarray}
Under the null hypothesis, equation (\ref{equation: relationship between observed responses with misreporting}) sustains a unique solution $\theta^*\in\Omega$. We can rewrite equation (\ref{equation: relationship between observed responses with misreporting}) as $J+2$ moment conditions. One of the moment conditions is  redundant for the purpose of estimation because the left hand side of the equations sums to 1. Thus, we have
\begin{eqnarray}\label{eq: vector of moment conditions}
\mathbb{E}[\psi(D, \theta^*)]=0,
\end{eqnarray}
where $\psi\equiv(\psi_0\hspace{0.2cm}\psi_1\hspace{0.2cm}\cdots \psi_{J})^{\prime}$.
Let $c_0\equiv \Pr(t=0)$ and $c_1\equiv\Pr(t=1)$.
The conditions $\psi_j(\cdot, \cdot)$ are defined as follows.
\begin{eqnarray}\label{eq: J+2 population moment conditions}
\psi_0(D;\theta)&=&(1-\delta)\frac{1-p_1}{1-p_0}\left(\frac{(1-t)}{c_0}\cdot
\mathbbm{1}(Y=0)-\frac{p_0}{J+1}\right)-\frac{t}{c_1}\cdot
\mathbbm{1}(Y=0)+\frac{p_1}{J+2},\nonumber\\
\psi_j(D; \theta)&=&\frac{(1-p_1)}{1-p_0}\left(\frac{\delta(1-t)}{c_0}\cdot
\mathbbm{1}(Y=j-1)+\frac{(1-\delta)(1-t)}{c_0}\cdot
\mathbbm{1}(Y=j)-\frac{p_0}{J+1}\right)\nonumber\\
&-&\frac{t}{c_1}\cdot
\mathbbm{1}(Y=j)+\frac{p_1}{J+2}, j=1, 2, \cdots, J,\nonumber\\
\psi_{J+1}(D, \theta)&=&\frac{\delta(1-p_1)}{1-p_0}\left(\frac{(1-t)}{c_0}\cdot
\mathbbm{1}(Y=J)-\frac{p_0}{J+1}\right)-\frac{t}{c_1}\cdot
\mathbbm{1}(Y=J+1)+\frac{p_1}{J+2}.
\end{eqnarray}
After dropping one redundant moment, the GMM estimator of $\theta^*$ is
\begin{eqnarray}
\hat{\theta}&=&\arg\min\nolimits_{\theta\in \Omega}\bar{\psi}^{\prime}(D;\theta)\widehat{W}\bar{\psi}(D;\theta),
\end{eqnarray}
where $\widehat{W}$ is the optimal weighting matrix and $\bar{\psi}$ is the sample analog of $\psi$. For example, the sample analogy of $\bar{\psi}_0$ is
\begin{eqnarray*}
\bar{\psi}_0(D;\theta)=(1-\delta)\frac{1-p_1}{1-p_0}\left(\frac{1}{n}\sum_{i=1}^n\frac{(1-t_i)}{\hat{c}_0}\cdot
\mathbbm{1}(Y_i=0)-\frac{p_0}{J+1}\right)-\frac{1}{n}\sum_{i=1}^n\frac{t_i}{\hat{c}_1}\cdot
\mathbbm{1}(Y_i=0)+\frac{p_1}{J+2},
\end{eqnarray*}
where $\hat{c}_0=\sum_{i=1}^n\mathbbm{1}(t_i=0)/n$ and $\hat{c}_1=\sum_{i=1}^n\mathbbm{1}(t_i=1)/n$.

Let $T_n$ be the $J$-statistic
\begin{eqnarray}
T_n\equiv n\cdot \left(\bar{\psi}^{\prime}(D;\hat{\theta})\widehat{W}\bar{\psi}(D;\hat{\theta})\right)\xrightarrow{{\enskip d \enskip}}\chi^2(J-2).
\end{eqnarray}
Note that our test nests the case where there is no misreporting in at least one group.
A rejection of the null hypothesis implies that at least one of the Assumptions \ref{asp: random assignment} and \ref{asp: misreporting probability} is violated, although we cannot distinguish which assumption is violated. If Assumption \ref{asp: random assignment} is rejected, then the randomization assignment is not successful, or the respondents in the two groups have different preferences over the non-sensitive questions, or both. If Assumption \ref{asp: misreporting probability} is rejected, then we mis-specify the respondents' misreporting strategies. When the null hypothesis is rejected for LE, the mean difference, OLS, NLS, MLE, and GMM would provide biased estimates.

Our estimation and test procedures based on the unconditional moment conditions can be readily applied to conditional moment conditions. When Assumption \ref{asp: random assignment} only holds conditional on the vector of respondents' characteristics $Z$, the moment conditions in equation (\ref{eq: vector of moment conditions}) can be written as $\mathbb{E}[\psi(D, \theta^*)|Z=z]=0$. The estimation and testing procedures are then based on the unconditional moments $\mathbb{E}[\psi(D, \theta^*)\cdot z]=0$. Furthermore, model parameters $\delta, p_0$, and $p_1$ can also depend on $Z$ and testing can be conducted following  existing methods (e.g., \citealt{andrews2017inference}). 

\subsection{A modified version of LE\label{section: modified LE}}
The modified version of LE proposed in
\cite{coffman2017size} is often used to justify the reliability (no misreporting or measurement error) of the survey data. The main idea is
to compare the direct responses to the sensitive question to the result derived from LE (for example, see \cite{cantoni2019protests} and \cite{chen2019impact}). Under the modified LE, we use $X\in\{0,1\}$ to denote the control group's direct response to the sensitive question. The probability $\Pr(X=1)$ represents the observed probability of the direct response to the sensitive question. The widely used justification for truthful reporting is based on the condition
\begin{eqnarray}\label{equation: justification of no misreporting}
\Pr(X=1)=\mathbb{E}(Y_1)-\mathbb{E}(Y_0).
\end{eqnarray}
The justification relies on the claim below.

\smallskip
\noindent \textbf{Claim:} \textit{If
the condition in equation (\ref{equation: justification of no misreporting})
holds, respondents reveal their true preference in answering the survey questions.}
\bigskip

To simplify our discussion, we assume that the misreporting behavior in LE is the same for the control and treatment groups, $p_1=p_0=p$. Corollary \ref{corollary: mean difference with misreporting} implies that $\mathbb{E}(Y_1)-\mathbb{E}(Y_0)=\delta+p(1-2\delta)/2$ under Assumptions \ref{asp: random assignment} and \ref{asp: misreporting probability}. When misreporting exists in respondents' answers to the sensitive question, we use
$q_1\equiv \Pr(X=0|X^*=1)$ and $q_0\equiv\Pr(X=1|X^*=0)$ to denote the reporting errors for respondents with and without the sensitive information, respectively, where $0\leq q_1, q_0\leq 1$. Note that we still allow
respondents to answer the sensitive question truthfully by setting $q_1=0$ and $q_0=0$. The probability of an affirmative answer to the sensitive question in the control group is
\begin{eqnarray*}
\Pr(X=1)=\sum\nolimits_{i\in\{0,1\}}\Pr(X=1|X^*=i)\Pr(X^*=i)=(1-q_1)\delta+q_0(1-\delta),
\end{eqnarray*}
and equation (\ref{equation: justification of no misreporting}) is then equivalent to
\begin{eqnarray}\label{equation: criterion of modified LE}
(1-q_1)\delta+q_0(1-\delta)=\delta+p(1-2\delta)/2.
\end{eqnarray}
We summarize this result in the following corollary of Proposition \ref{proposition: model restrictions under randomization}.
\begin{coro} \label{corollary: modified LE}
Suppose Assumptions \ref{asp: random assignment} and \ref{asp: misreporting probability} hold, the condition $\Pr(X=1)=\mathbb{E}(Y_1)-\mathbb{E}(Y_0)$ is equivalent to $(1-q_1)\delta+q_0(1-\delta)=\delta+p(1-2\delta)/2$.
 \end{coro}
Corollary \ref{corollary: modified LE} implies that truthful reporting is \textit{a sufficient} rather than a necessary condition for the testable condition in equation (\ref{equation: justification of no misreporting}). Without reporting errors, we have $(q_1,q_0, p)=(0,0,0)$, both the direct responses to the sensitive question and the mean difference identify the probability of sensitive information $\delta$. However, for a given $\delta$, there are infinitely many triplets $(q_1,q_0, p)\neq(0,0,0)$ such that equation (\ref{equation: justification of no misreporting}) holds.
Therefore, the claim that equation (\ref{equation: justification of no misreporting}) implies truthful reporting is problematic. Even respondents do not lie in LE ($p=0$), the claim may still be problematic because
$(1-q_1)\delta+q_0(1-\delta)=\delta$ does not imply $q_1=0$ and $q_0=0$ without further information about respondents' reporting strategy in answering the sensitive question.

As a result, the modified LE approach relies on the validity of LE, which requires Assumption \ref{asp: random assignment} or \ref{asp: misreporting probability} to hold. In Section \ref{section: application of the test}, we show that more than half of the cases we tested violate at least one of the assumptions, and the validity of LE is undermined in these cases. In addition, the comparison between LE and direct responses cannot be used to justify the reliability of survey data without further information about the reporting strategy of respondents. To address these possible issues and to check the reliability of the data, we can first test the validity of LE using the method proposed in Section \ref{subsection: test misreporting behavior}. A rejection of LE implies that the modified LE cannot be used to verify the existence of measurement error in the survey data. It is also impossible to recover sensitive information from the data by using LE. An acceptance of LE, together with the estimated parameters of $p$ and $\delta$, is still insufficient for the modified LE approach to work because both $q_1$ and $q_0$ are latent parameters and cannot be recovered directly from the survey data.

\subsection{Application of the Test\label{section: application of the test}}
In this section, we apply our testing procedure to five recently published peer-reviewed articles with publicly available data,\footnote{These are the only peer-reviewed articles with publicly available data that we can find.} including \cite{cantoni2019protests}, \cite{chen2019impact}, \cite{coffman2017size}, \cite{muralidharan2016building}, and \cite{neggers2018enfranchising}.
\cite{cantoni2019protests} and \cite{chen2019impact} rely on the modified LE to justify the reliability of the survey data. \cite{muralidharan2016building} and \cite{neggers2018enfranchising} directly estimate respondents' sensitive information by using LE. Specifically, \cite{coffman2017size} demonstrate  the existence of misreporting by comparing respondents' direct response to sensitive questions and responses from a LE. The empirical results of the these five articles are fully or partially based on the validity of LE.

\begin{table}
\centering
\caption{Results of estimation and testing}
\vspace{-0.8cm}
\label{table: testing the assumption of LE}
\begin{center}
\scalebox{0.65}{\begin{threeparttable}
\begin{tabular}{llllllllll}
\hline\hline
  & & &\multicolumn{3}{c}{GMM estimate} &&\multicolumn{3}{c}{testing results} \\
\cline{4-6} \cline{8-10}\vspace{-0.4cm}
\\
sensitive question&\tabincell{c}{sample\\ size}  & \tabincell{c}{mean \\ difference}& $\delta$& $p_0$ & $p_1$ &&\tabincell{c} {$p_0\neq p_1$}&\tabincell{c} {$p_0= p_1$}&\tabincell{c} {$p=0$}\\
\hline
favorable view of CCP &1576 & 0.057 & 0.039 & 0.081 & --- &&$\checkmark$ &$\checkmark$& \xmark  \\
& & (0.057) & (0.073) & (0.047) &  \\
consider self Hong Kongese&1576 & 0.816 &---&---&---
&& \xmark&\xmark &\xmark \\
& & (0.048) \\
support for HK independence &1576 & 0.521 & 0.519 & 0.018 & --- && $\checkmark$& $\checkmark$& \xmark \\
& & (0.052) & (0.057) & (0.089) &  \\
\tabincell{l}{support violence in pursuit of\\  HK's political rights} &1576 & 0.389 & 0.423 & --- & --- && $\checkmark$& $\checkmark$& $\checkmark$ \\
& & (0.048) & (0.052) & &  \\
\hline
\tabincell{l}{I completely trust the \\central government of China}&1807 & 0.290 &---&---&---
&& \xmark&\xmark &\xmark \\
& & (0.038)  \\
 \hline
\tabincell{l}{do you consider yourself to be heterosexual?}&2516 & 0.894 &0.891&0.000&0.045
&& $\checkmark$&\xmark&$\dagger$ \\
& & (0.036) & (0.049) & (0.001) & (0.032) \\
\tabincell{l}{are you sexually attracted to members of \\the same sex?}&2516 & 0.258 &---&---&---
&& \xmark&\xmark &$\dagger$ \\
& & (0.035)  \\
\tabincell{l}{ have you had a sexual experience with \\ someone of the same sex?}&2516
& 0.545 &---&---&---
&& \xmark&\xmark &\xmark  \\
& & (0.040)  \\
\tabincell{l}{do you think marriages between gay and \\ lesbian couples should be recognized by \\ the  law as valid, with the same rights \\as heterosexual marriages?}&2516 & 0.986&---&---&---
 && \xmark&\xmark &\xmark \\
& & (0.032)  \\
\tabincell{l}{ would you be happy to have an openly \\lesbian, gay, or bisexual manager at work?}&2516 & 0.912 &---&---&---
&& $\dagger$&$\dagger$ &\xmark \\
& & (0.038)  \\
\tabincell{l}{do you believe it should be \\ illegal to discriminate  in hiring\\ based on someone's sexual orientation?}&2516 & 0.806
&---&---&---
 && \xmark&\xmark &\xmark  \\
& & (0.031)  \\
\tabincell{l}{do you believe lesbians and gay men \\ should be allowed to adopt children?}&2516 & 0.878 & --- & --- & --- && \xmark& \xmark& \xmark  \\
& & (0.034)  \\
\tabincell{l}{do you think someone who is \\ homosexual can change their sexual\\ orientation if they choose to do so?}&2516 & 0.186 & 0.187 & --- & --- && $\checkmark$& $\checkmark$& $\checkmark$ \\
 && (0.034) & (0.031) &  &  \\
\hline
\tabincell{l}{w/o smartcards system, members\\ of this household have been asked\\ by officials to lie about the\\ amount of work they did on NREGS} &917 & 0.044 &0.014 & --- & --- && \xmark & $\checkmark$ & $\checkmark$ \\
 && (0.059) & (0.072) & & \\
 \tabincell{l}{w/o smartcards system, members\\ of this household have been given\\ the chance to meet with the CM of AP\\ to discuss problems with NREGS?}&897 & 0.174 & 0.182 & 0.063 & 0.003 && $\checkmark$ & \xmark & \xmark \\
& & (0.062) & (0.090) & (0.102) & (0.078) \\
\tabincell{l}{w/ smartcards system, members\\ of this household have been asked\\ by officials to lie about the\\ amount of work they did on NREGS}&2300 & 0.106 &---&---&---
&&\xmark&\xmark&\xmark \\
& & (0.036)  \\
\tabincell{l}{w/ smartcards system, members\\ of this household have been given\\ the chance to meet with the CM of AP\\ to discuss problems with NREGS?}&2276 & 0.154 & 0.140 & --- & --- && $\checkmark$ & $\dagger$ & $\checkmark$ \\
 && (0.037) & (0.053) &  &  \\
\hline
\tabincell{l}{treated voters differently by religion/caste} &3833&	0.242&	0.248&---	&---	&&	
$\dagger$&$\checkmark$& $\checkmark$\\
&& (0.027) & (0.029) &  &  \\
\tabincell{l}{attempted to influence voting}&3850	&0.142	&---&---&---		&&\xmark&\xmark&\xmark\\
& & (0.038)  \\
 \hline\hline
\end{tabular}
\begin{tablenotes}
\item Note: The five panels (from top to bottom) are results for \cite{cantoni2019protests}, \cite{chen2019impact}, \cite{coffman2017size}, \cite{muralidharan2016building}, and \cite{neggers2018enfranchising}, respectively.
\item \xmark, $\dagger$, and $\checkmark$ indicate $p<0.05$, $p<0.1$, and $p>0.1$, respectively. Standard errors in parentheses are bootstrapped $1000$ times. When only $\hat{\delta}$ is provided, then $p=0$; if both $\hat{\delta}$ and $\hat{p}_0$ are provided, then $p_1=p_0$.
\end{tablenotes}
\end{threeparttable}
}
\end{center}
\end{table}
The setting of the surveys in these papers is the same as in our model. There are five ($J=5$) in \cite{muralidharan2016building} and four ($J=4$) nonsensitive questions in the other articles. To accommodate a flexible misreporting structure, we estimate and test the model under three alternative misreporting specifications:  (1) $p_0\neq p_1$, (2) $p_0=p_1$, and (3) $p_0=p_1\equiv p=0$. In theory, the first test nests the second, which further nests the third condition.
Nevertheless, the testing results may violate the nesting relationships in finite samples.
We also estimate the mean difference of $\mathbb{E}(Y_1)-\mathbb{E}(Y_0)$ following the existing literature. The standard errors are computed by bootstrapping $1,000$ times. In Table \ref{table: testing the assumption of LE}, we present the sensitive survey questions and the corresponding results of estimation and testing.\footnote{In all three scenarios, we have to drop one redundant moment condition. We conduct the estimation and hypothesis testing by dropping different moment conditions; the results are qualitatively similar. For each sensitive question, we present the results with the smallest $p$-value when a moment condition is dropped.} The estimates are presented only if the model cannot be rejected in at least one of the three specifications. If there is no rejection in at least two specifications, then we present the estimates for the more restrictive specification.\footnote{Recall that under Assumptions \ref{asp: random assignment} and \ref{asp: misreporting probability}, the mean difference also identifies the parameter $\delta$ if  $p_1=0$ and $\mathbb{E}(Y_0)=J/2$. We test the hypothesis $\mathbb{E}(Y_0)=J/2$ and find that it is rejected for all but the first question in \cite{coffman2017size}.}

We find that more than half of the test results indicate that the assumptions of LE are rejected. Each article has at least one rejection.  At the 5\% significance level, we reject the null hypothesis for 10, 11, and 12 questions out of 19, accounting for 53\%, 58\%, and 63\% for specifications (1), (2), and (3), respectively. At the 10\% significance level, the percentage of cases being rejected is
63\%, 68\%, and 74\% in specifications (1), (2), and (3), respectively. The model is rejected for 42\% (8 out of 19) and 63\% (12 out of 19) of the questions
at the 5\% and 10\% significance level, respectively, in all three specifications.  The estimates of sensitive information based on LE are biased when the assumptions are rejected. When the condition in specification (2) or (3) is rejected, the mean difference approach produces biased estimates and GMM can be applied to estimate the treatment effect $\delta$. When the condition in specification (1) is rejected, the underlying assumptions in GMM estimation also are violated and the treatment effect $\delta$ cannot be recovered from the data generation process.

We observe that when LE is not rejected, the reporting errors are estimated to be small and insignificant. All of the cases with valid LE included in our sample have statistically insignificant reporting errors, suggesting that misreporting may play an important role in the validity of LE.  The level of sensitivity of a survey question is not necessarily correlated with a model rejection. For example, in \cite{cantoni2019protests}, ``consider self Hong Kongese" is probably less sensitive than ``support violence in pursuit of HK's political rights", but the former is rejected and the latter is not. In \cite{coffman2017size}, ``do you think someone who is homosexual can change their sexual orientation if they choose to do so?" probably is as sensitive as ``do you consider yourself to be heterosexual". We fail to reject the model for the former question but reject the latter except for the scenario $p_1\neq p_0$.

We also conduct a test of the five articles under an alternative assumption:  respondents misreport only if their truthful
responses disclose privacy, i.e., when their truthful response is the outcome $J + 1$ in the treatment group. The empirical results presented in Table \ref{table: testing the assumption of LE alt} are qualitatively
similar to our findings above, with a slightly larger proportion of rejection. The details of the model and testing are summarized in
Appendix \ref{appendix: an alternative misreporting strategy}. The robustness of our testing results to specifications of measurement error implies that Assumption \ref{asp: random assignment} is more likely to cause the failure of LE than Assumption \ref{asp: misreporting probability}.

\section{Multiple Response Technique}
When equation (\ref{equation: relationship between observed responses with misreporting}) is rejected for a dataset, LE is no longer applicable to estimating the probability of positive responses to the sensitive question. The modified LE also relies on the latent reporting strategies of respondents and is subject to the restrictions in Proposition \ref{proposition: model restrictions under randomization}. Motivated by these issues, we propose a novel survey approach to recovering
the distribution of sensitive information and the latent reporting strategies of respondents.
In this approach, we treat sensitive information as respondents' unobserved heterogeneity and recover it from a series of responses to direct questions.

\subsection{The design of survey}
We maintain the assumption that there are $n$ respondents with a vector of characteristics $Z$.
All respondents answer the same set of survey questions and random assignment to groups is not required. In the survey, we ask three or more yes-no questions that are related to the sensitive information. Without loss of generality, we assume that there are three questions. Let $X_j\in\{0, 1\}$ be the random variable that describes the answer to the $j$-th question for $j=1, 2, 3$. The responses $X_j$ depend on the latent true preference
toward the sensitive information, denoted by $X^*\in\{0,1\}$. We aim to recover the conditional probability distribution of $X^*$, or $\Pr(X^*=1|Z)$, from the joint distribution of sensitive question answers $X_1, X_2, X_3$ and respondent characteristics $Z$. Our framework allows $X^*$ to take more than two values when the number of possible answers to each question is not less than the values $X^*$ takes. Three questions are sufficient for our purpose and more questions can be accommodated easily in our framework.

Suppose we are interested in the size of the LGBT population. The three yes-no survey questions related to LGBT sex orientation used in \cite{coffman2017size} are: (1) do you consider yourself to be heterosexual? (2)
are you sexually attracted to members of the same sex? (3) have you had a sexual experience with someone of the same sex?
Below we discuss the choice of these survey questions and respondents' characteristics,
guided by the identification strategy in the next section.


\subsection{The identification strategy}
Let $\Pr\left( X_{1}, X_2, X_3, Z\right)$ be the observed  joint distribution of three binary answers ($X_{1}, X_2, X_3$) to the survey questions and respondents' characteristics ($Z$). The objective of identification is to recover the conditional probability $\Pr(X^*=1|Z)$ from $\Pr\left( X_{1}, X_2, X_3, Z\right)$. The probability of respondents' answer to the $j$-th question conditional on their sensitive information and characteristics,  $\Pr(X_j=1|X^*,Z), j=1, 2, 3$, also can be identified in our framework. Our identification strategy is to treat the sensitive information as an unobserved heterogeneity of respondents, then to recover the distribution of the  heterogeneity from its multiple measurements, based on the methodology
of measurement errors proposed in \cite{hu2008identification}.

We first discuss the assumptions required to achieve identification.
\begin{asp}\label{assumption: conditional independence}
Given respondents' characteristics $Z$ and latent true response $X^*$ to the sensitive question, their responses to the three questions are independent.
\begin{equation*}
\Pr\left( X_{j}|X_{i}, X_k, X^{\ast },Z\right) =\Pr\left(X_{j}|X^{\ast },Z\right).
\end{equation*}
\end{asp}
Under Assumption \ref{assumption: conditional independence}, the responses to the three questions are allowed to be correlated through
respondents' sensitive information and their characteristics.
 The response to a sensitive question is determined by three components: the characteristics $Z$, the sensitive information $X^*$, and some random factors summarized by $\epsilon$. Without loss of generality, we express $X_j$ as $X_j=h_j(Z, X^*, \epsilon_j)$, where $h_j(\cdot)$ is an unknown function and $\epsilon_j$ is allowed to be correlated with both the characteristics $Z$ and the sensitive information $X^*$. Assumption \ref{assumption: conditional independence} states that after controlling for $Z$ and $X^*$, the remaining information in $\epsilon_j$ is mutually independent, i.e., $\epsilon_i \perp \epsilon_j| Z, X^*$.
 In the example of the LGBT population above, if a respondent's sensitive information is gay ($X^*=1$),
then the responses to the three questions ``heterosexuality", ``same-sex attraction", and ``same-sex sexual experience" could be correlated, because the respondent may respond strategically. Assumption \ref{assumption: conditional independence} requires that such correlation is only through
the demographics and  sexual orientation.

Assumption \ref{assumption: conditional independence} rules out the possibility that other unobserved heterogeneity
of respondents may affect their responses to the survey questions. Such restrictions can be alleviated by carefully choosing respondents' characteristics $Z$. For example, inclusion of respondents' religion  could successfully control respondents' religious effects on their responses. Although the assumption cannot be tested empirically, we evaluate the sensitivity of this assumption on our survey technique using Monte Carlo simulations. We find that our method is not sensitive to the assumption when the correlation is weak.

Our identification procedure begins with the relationship between the observed joint probabilities $\Pr\left( X_{1}, X_2, X_3|Z\right)$ and the model primitives $\Pr(X_j|X^*,Z)$ and $\Pr(X^*|Z)$. Under Assumption \ref{assumption: conditional independence}, we apply the law of total probability to $\Pr\left( X_{1}, X_2, X_3|z\right)$ for a given $z$,
\begin{eqnarray}\label{equation: the law of total probability with three variables}
\Pr\left( X_{1}, X_2, X_3|z\right) =\sum\nolimits_{X^*\in\{0,1\}}\Pr\left( X_{1}| X^*,z\right)\Pr\left( X_{2}| X^*,z\right)\Pr\left( X_3,X^*|z\right).
\end{eqnarray}
The equation above provides seven independent restrictions to seven unknown parameters.\footnote{$X_1, X_2$ and $X_3$ are binary, $\Pr\left( X_{1}, X_2, X_3|z\right)$ provides $2\times2\times2=8$ equations. The summation of the equations is one, thus there are seven  independent equations. $\Pr\left( X_{j}| X^*,z\right)$ contains one parameter for each $x^*$. The number of
of parameters in $\Pr\left( X_{1}| X^*,z\right)$ and $\Pr\left( X_{2}| X^*,z\right)$ is four. Similarly, there are three parameters in $\Pr\left( X_3,X^*|z\right)$. The total number of parameters is seven.} Nevertheless, these equations are nonlinear, and the unknown parameters may not be identifiable. Following \cite{hu2008identification}, these restrictions can be written in matrix form. Without loss of generality, we fix a value of $X_2=x_2\in\{0,1\}$ such that the joint distribution $\Pr\left( X_{1}, X_2, X_3|z\right)$ can be written as a matrix. For the given $x_2$, the matrix form of Equation (\ref{equation: the law of total probability with three variables}) is:
\begin{eqnarray}\label{equation: joint distribution of three observations}
M_{X_1, x_2,X_3|z}=M_{X_1|X^*, z}D_{x_2|X^*,z} M^{\prime}_{X_3, X^*|z},
\end{eqnarray}
where the matrices are defined as
\begin{eqnarray*}
\left(M_{X_1, x_2,X_3|z}\right)_{i,j} &=&\Pr(X_1=i, x_2, X_3=j|z),\nonumber\\
\left(M_{X_{1}|X^{\ast },z}\right)_{i,k} &=&\Pr(X_1=i|X^*=k, z),\nonumber\\
\left(M_{X_{3},X^{\ast }|z}\right)_{k,j} &=&\Pr(X_3=k,X^*=j| z),\nonumber\\
D_{x_{2}|X^{\ast },z}&=&\hbox{diag}\big[\Pr(x_2|X^*=0,z) \hspace{0.3cm} \Pr(x_2|X^*=1,z)\big].
\end{eqnarray*}
The matrix form (\ref{equation: joint distribution of three observations}) only provides three equations for the seven unknowns parameters. The joint distribution $\Pr\left( X_{1}, X_2, X_3|z\right)$ contains further identification information on $\Pr\left( X_{j}, X_{k}|z\right)$ and $\Pr(X_{j}|z)$, where the latter two distributions are dependent of the model primitives $\Pr(X_j|X^*, z)$ and $\Pr(X_k|X^*,z)$. We explore the identifying power of the joint distribution $\Pr\left( X_1, X_3|z\right)$ and rewrite $\Pr\left( X_1, X_3|z\right)$ in matrix form 
\begin{eqnarray}\label{equation: joint distribution of two observations}
M_{X_1, X_3|z}=M_{X_1|X^*, z} M^{\prime}_{X_3, X^*|z},
\end{eqnarray}
where the matrix $M_{X_1, X_3|z}$ is defined as
$\left(M_{X_{1},X_3|z}\right)_{k,l}=\Pr(X_1=k, X_3=l|z).$
Equations (\ref{equation: joint distribution of three observations}) and (\ref{equation: joint distribution of two observations})
summarize all of the model restrictions to the seven unknowns and allow us to identify the parameters under the additional assumption below.
\begin{asp}\label{assumption: invertibility, distinct eigenvalues, and ordering}
Given characteristics $z$, (i) the probability that respondents with sensitive information answer ``yes" (or ``no") to each of the three questions is different from the probability of those without sensitive information, or
$\Pr(X_j=1|X^*=0, z)\neq \Pr(X_j=1|X^*=1, z), \forall z,  j=1, 2, 3$; and (ii) respondents with sensitive information answer ``yes" with a larger or smaller probability than respondents without sensitive information, that is, the order of $\Pr(X_k=1|X^*=0,z)$ and $\Pr(X_k=1|X^*=1,z)$ is known for $k=1$ or $k=2$.
\end{asp}
Assumption \ref{assumption: invertibility, distinct eigenvalues, and ordering}(i)  requires that given characteristics $z$, a non-LGBT respondent responds positively (or negatively) to the three questions with a different probability than an LGBT respondent. For example, $\Pr(X_1=0|X^*=0, z) \neq \Pr(X_1=0|X^*=1, z)$ implies that the probability
that a non-LGBT respondent is against same-sex marriage differs from that of an LGBT respondent, after controlling for characteristics $z$. Assumption \ref{assumption: invertibility, distinct eigenvalues, and ordering}(i) holds for $j=1,3$ if and only if the matrix $M_{X_1, X_3|z}$ is invertible for $z\in \mathcal{Z}$. The invertibility of $M_{X_1, X_3|z}$ can be tested from the observed sample of $X_1$ and $X_3$ based on some existing methods, e.g., \cite{robin2000tests}. Similarly, we can test the invertibility of the matrix $M_{X_1, X_2|z}$ to verify whether part (i) holds for $j=2$. Assumption \ref{assumption: invertibility, distinct eigenvalues, and ordering}(ii) requires that $\Pr(X_k=1|X^*=0,z)$ and $\Pr(X_k=1|X^*=1,z)$ can be ordered for $k=1$ or $k=2$. For example, the assumption $\Pr(X_1=1|X^*=0,z)<\Pr(X_1=1|X^*=1,z)$ states that LGBT respondents are more likely to support same-sex marriage than non-LGBT respondents after for controlling their characteristics.

Under Assumption \ref{assumption: invertibility, distinct eigenvalues, and ordering}(i), we take the inverse of both sides of equation (\ref{equation: joint distribution of two observations}) and multiply it
from right to equation (\ref{equation: joint distribution of three observations}),
\begin{eqnarray}\label{equation: main identification equation}
M_{X_1, x_2,X_3|z}M^{-1}_{X_1, X_3|z}=M_{X_1|X^*, z}D_{x_2|X^*,z} M^{-1}_{X_1|X^*, z}.
\end{eqnarray}
The left-hand-side is a product of two observed matrices and the right-hand-side is an eigenvalue-eigenvector decomposition of the left-hand-side, with
$M_{X_1|X^*, z}$ and $D_{x_2|X^*,z}$ being the eigenvector and eigenvalue matrices, respectively.

The latent matrices $M_{X_1|X^*, z}$ and $D_{x_2|X^*,z}$ are identified from a unique decomposition of matrix $M_{X_1, x_2,X_3|z}M^{-1}_{X_1, X_3|z}$. It requires that: (a) the eigenvector matrix is normalized; (b) the two eigenvalues are distinct; and (c) eigenvalues or eigenvectors are correctly ordered. Requirement (a) holds without additional assumptions: the column sum of the matrix $M_{X_1|X^*, z}$ is one, dividing each element by its column sum  normalizes the eigenvector matrix. Requirement (b) is satisfied under Assumption \ref{assumption: invertibility, distinct eigenvalues, and ordering}(i) for $x_2$, i.e., $\Pr(X_2=1|X^*=1,z) \neq \Pr(X_2=1|X^*=0,z)$.  Assumption \ref{assumption: invertibility, distinct eigenvalues, and ordering}(ii) ensures that requirement (c) is met. If the inequality holds for $k=1$ and $k=2$, then the two columns of the eigenvector matrix and two eigenvalues can be ordered, respectively.
Once the matrices $M_{X_1|X^*, z}$ and $D_{x_2|X^*,z}$ are identified, we can recover the distribution of $X^*$, $\Pr(X^*|z)$ from \begin{eqnarray*}
\Pr(X_1|z)=\sum\nolimits_{X^*\in\{0,1\}}\Pr\left( X_{1}, X^*|z\right)=\sum\nolimits_{X^*\in\{0,1\}}\Pr\left( X_{1}|X^*, z\right)P(X^*|z),
\end{eqnarray*}
or
$
M_{X_1|z} =M_{X_1|X^*,z} M_{X^*|z}
$
in matrix form, where $M_{X^*|z}=\big[\Pr(X^*=0| z) \hspace{0.2cm}\Pr(X^*=1| z)\big]^{\prime}$ and $M_{X_1|z} =\left[\Pr(X_1=0| z)\hspace{0.2cm}\Pr(X_1=1| z)\right]^{\prime}$. Since $M_{X_1|X^*,z}$ is invertible, we have
\begin{eqnarray}\label{equation: estimation of the conditional probabilities of sensitive information}
M_{X^*|z}=M^{-1}_{X_1|X^*,z} M_{X_1|z} .\end{eqnarray}
The unconditional probability of sensitive information $\Pr(X^*)$ can be calculated by integrating $\Pr(X^*|z)$ over $z$.

We summarize the identification results in the theorem below.
\begin{thm}\label{theorem: results of nonparametric identification}
Under Assumptions \ref{assumption: conditional independence} and \ref{assumption: invertibility, distinct eigenvalues, and ordering}, the probabilities of sensitive information $\Pr(X^*|z)$, and the respondents' responses conditional on sensitive information $\Pr(X_j|X^*,z), j=1, 2, 3$ are uniquely determined by the joint distribution $\Pr(X_1, X_2, X_3|z)$.
\end{thm}
\begin{proof}
See Appendix.
\end{proof}
In some experimental settings where both our method and LE are applicable, one can compare the probability of sensitive information $\Pr(X^*|z)$ obtained by the two methods. If LE is rejected by using our test,
then such a comparison is meaningless. If LE is not rejected, then the two methods may yield different estimates because they are based on different sets of assumptions.

The intuition of the results in Theorem \ref{theorem: results of nonparametric identification} is that a latent variable can be recovered from its multiple measurements (see \cite{hu2017econometrics}). The response to each question $X_j$ provides partial information on the latent variable $X^*$; the joint distribution contains all of the information on the latent variable as well as measurement errors. The main idea behind identification is to use additional information to recover the latent parameters, similar to the use of instrumental variables in a linear regression.

Assumptions \ref{assumption: conditional independence} and \ref{assumption: invertibility, distinct eigenvalues, and ordering} are sufficient for identification. They provide clear guidance for researchers to choose survey questions and respondents' characteristics. The result in Theorem \ref{theorem: results of nonparametric identification} is for cases with binary sensitive information, and it also applies when the number of possible responses to sensitive information is greater than two.

Our method has several advantages over the existing methods of eliciting sensitive information. First, we are able to recover sensitive information without random assignment of respondents. The testing results in Section \ref{section: application of the test} indicate that the assumption of random assignment and/or no impact of the sensitive question may fail in LE. In our design, respondents answer the same set of questions and the impacts of assignment to different groups on estimation are ruled out.

Second, our approach allows for arbitrary misreporting, including no misreporting,  by respondents. Using mean difference or regression approaches, the existence of misreporting may bias the LE method. The GMM potentially may fix the misreporting issue under limited circumstances. Unfortunately, GMM further relies on the linear misreporting error assumed in Assumption \ref{asp: misreporting probability}, which may be invalid in some surveys. If the misreporting errors enter the responses nonlinearly, then identification of LE would be much more complicated and likely to fail. By contrast, as shown in our identification strategy, we impose no restrictions on misreporting.
Moreover, when respondents are directly asked about their sensitive information  (e.g., in \cite{coffman2017size} respondents are asked directly whether they are heterosexual), we can recover the magnitude of misreporting from respondents' answers. We will explain this in detail in the application in Section \ref{subsection: estimating sexual orientation}.

Third, we can obtain sensitive information at a disaggregated level. \cite{bertrand2017field} point out that a major disadvantage of LE, as well as other methods of eliciting sensitive information, is that sensitive information only can be recovered at the population level. We advance the literature by decomposing any result for the population into different groups,
characterized by the unobserved level of sensitive information.  For example, a study of ``supporting same-sex marriage" using LE (suppose that LE is valid), as in \cite{coffman2017size}, only gives a supporting rate at an aggregated level: a weighted average of those who are friendly to the LGBT population and those who are not. We can recover the rates for the two individual groups by treating the proportion of respondents with sensitive information as the weight.

Moreover, researchers often are interested in the dependence of sensitive information on covariates. We can parametrize the conditional probability  $\Pr(X^*=1|z)=g(z;\theta)$ and estimate the marginal effect of $z$ at a given value of $z_i$. Similarly,
we can obtain a prediction on the probability of sensitive information for a respondent with a characteristic $z_i$ by $\Pr(X^*=1|z_i;\theta)$.

 \subsection{Estimation and simulation}
 The model parameters identified in Theorem \ref{theorem: results of nonparametric identification} can be estimated either parametrically or nonparametrically. When covariates $Z$ are discrete, it is convenient to estimate the model primitives nonparametrically for each possible value of $Z$. The first approach is to follow
the constructive identification procedure and apply the eigenvalue-eigenvector decomposition for estimation. The main advantage of this approach is that the estimator is global and involves no optimization and iteration. One possible drawback is that the estimated probability might be outside the interval $[0,1]$ due to finite sample properties. To address this, we propose an extreme estimator, minimizing a matrix norm for the difference between two sides of equation (\ref{equation: main identification equation}), where the latent probabilities are constrained to $[0,1]$. When covariates $Z$ are continuous, we estimate the model primitives using MLE. Our first step is to parametrize the latent conditional probabilities, e.g., $\Pr(X^*=1| z) \equiv g(z; \theta)$ with $g(\cdot; \theta)$ being a logistic function. The next step is to maximize the likelihood function based on equation (\ref{equation: main identification equation}). The details are in Appendix \ref{section: estimation methods}.

We conduct Monte Carlo simulations to demonstrate that the proposed estimators perform satisfactorily even for a sample with modest size ($N=500$). The estimating method performs better for discrete covariate than for continuous covariate. We also check the sensitivity of our estimates to the conditional independence Assumption \ref{assumption: conditional independence}. Our findings show that our estimation is robust when $X_1, X_2$, and $X_3$ are weakly correlated (e.g., the correlation coefficients are less than 0.1) conditional on $X^*$ and $Z$. Detailed simulation procedures and results are in Appendix \ref{section: Monte Carlos simulations}.

\section{Estimating LGBT-Related Information}

Data on LGBT-population and LGBT-related sentiment play an important role in a wide range of
topics of research and policy: e.g., discrimination in the labor market; sexually transmitted diseases
and related policies; the demand for children; educational investment; and household labor supply.\footnote{Please
see \cite{coffman2017size} for a detailed discussion.} Unfortunately, such survey data are very likely subject to misreporting
and the results from different surveys vary significantly (\cite{gonsiorek1995definition}). It is unclear in the literature how
researchers could obtain reliable estimates based on the LGBT-related data containing measurement error.

In this section, we apply our multiple response technique (MRT) to estimate respondents' sexual orientation and LGBT-related sentiment. We then provide some novel findings that are obscured with LE or other existing methods of survey.
 The data come from \cite{coffman2017size} where respondents answer a series of sensitive questions on sexual orientation and LGBT-related sentiment.

\subsection{Sexual orientation\label{subsection: estimating sexual orientation}}
First we estimate the sexual orientation of respondents. Let $X^*=1$ denote  latent sexual orientation being non-heterosexuality.
We then use
the responses to the three sexuality-related sensitive questions in  \cite{coffman2017size} as measurements $X_j, j=1, 2, 3$,
\begin{enumerate}
 \setlength\itemsep{0.001cm}
\item Do you consider yourself to be heterosexual?
\item Are you sexually attracted to members of the same sex?
\item Have you had a sexual experience with someone of the same sex?
\end{enumerate}
We use $X_j=1$ to represent an affirmative answer to these above questions.
 To examine the possible dependence of sexual orientation on respondents' demographics,
 we choose the covariates $Z$ to be gender, race, religion, politics, and age, based on the findings in \cite{coffman2017size} that these demographics affect respondents' reporting behaviors. Assumption  \ref{assumption: conditional independence} requires that given a respondent's sexual orientation and demographics, responses to the three questions above are independent. Assumption \ref{assumption: invertibility, distinct eigenvalues, and ordering}(i) requires that the proportion of respondents answering ``yes" to any of the three questions varies across the two groups of respondents with different sexual orientations.  The matrices $M_{X_2,X_3|z}$ and $M_{X_1, X_2|z}$ are full rank for all covariates based on the test method in  \cite{robin2000tests}, implying that Assumption \ref{assumption: invertibility, distinct eigenvalues, and ordering}(i) holds. Assumption \ref{assumption: invertibility, distinct eigenvalues, and ordering}(ii) requires that the heterosexual respondents are more likely to provide
 an affirmative answer than their counterparts to the first question (``heterosexual") given their demographics, i.e., $\Pr(X_1=1|X^*=0,z)>\Pr(X_1=1|X^*=1,z)$.\footnote{The response to the first question $X_1$ is used eigenvalues in Section \ref{section: estimation methods}.}
The model primitives and their standard errors are estimated using the extreme estimator and bootstrapped for 1,000 times, respectively.  We present the results of estimation in Table \ref{table: estimation results of LGBT probability} and further depict the estimated proportion of non-heterosexuality and their 95\% confidence intervals in the left panel of Figure \ref{fig:effects of demographics on sexuality and sentiment}.
\begin{table}
\centering
\caption{Results of estimation: Sexual orientation}
\vspace{-0.6cm}
\label{table: estimation results of LGBT probability}
\begin{center}
\scalebox{0.70}{\begin{threeparttable}
\begin{tabular}{llllllllllllllllll}
\hline\hline
\cline{3-17}\\
& &\multicolumn{2}{c}{gender} &&\multicolumn{2}{c}{race} &&\multicolumn{2}{c}{religion}&&\multicolumn{2}{c}{politics}&&\multicolumn{3}{c}{age} \\
\cline{3-4} \cline{6-7}\cline{9-10}\cline{12-13} \cline{15-17}
\\
parameter&overall&male & female && white & black && christian& no reli.&& dec.& rep.&&$<31$ &31-50&$>50$ \\
\hline
$\Pr(X^*=1)$  &0.113  & 0.074 & 0.170 &  & 0.119 & 0.094 &  & 0.077 & 0.141 &  & 0.149 & 0.047 &  & 0.132 & 0.102 & 0.024 \\
 &(0.019) & (0.011) & (0.020) &  & (0.013) & (0.191) &  & (0.033) & (0.019) &  & (0.018) & (0.088) &  & (0.015) & (0.020) & (0.118) \\
 \hline
$\Pr(X_1=1|1)$ &0.293& 0.351 & 0.260 &  & 0.285 & 0.492 &  & 0.394 & 0.238 &  & 0.256 & 0.473 &  & 0.287 & 0.339 & 0.000 \\
 &(0.048)& (0.081) & (0.060) &  & (0.051) & (0.225) &  & (0.104) & (0.062) &  & (0.057) & (0.186) &  & (0.058) & (0.097) & (0.125) \\
$\Pr(X_1=1|0)$ &0.963& 0.965 & 0.959 &  & 0.965 & 0.956 &  & 0.962 & 0.967 &  & 0.955 & 0.970 &  & 0.955 & 0.980 & 0.971 \\
 &(0.031)& (0.031) & (0.032) &  & (0.031) & (0.043) &  & (0.032) & (0.032) &  & (0.032) & (0.033) &  & (0.031) & (0.032) & (0.210) \\
 \hline
$\Pr(X_2=1|1)$&1.000&  1.000 & 0.984 &  & 1.000 & 1.000 &  & 0.997 & 1.000 &  & 1.000 & 1.000 &  & 1.000 & 0.885 & 1.000 \\
&(0.018) & (0.016) & (0.028) &  & (0.015) & (0.161) &  & (0.050) & (0.013) &  & (0.018) & (0.063) &  & (0.004) & (0.080) & (0.133) \\
$\Pr(X_2=1|0)$ &0.029& 0.019 & 0.046 &  & 0.028 & 0.047 &  & 0.024 & 0.027 &  & 0.038 & 0.037 &  & 0.033 & 0.024 & 0.015 \\
 &(0.032)& (0.032) & (0.034) &  & (0.032) & (0.165) &  & (0.045) & (0.033) &  & (0.033) & (0.042) &  & (0.033) & (0.033) & (0.235) \\
 \hline
$\Pr(X_3=1|1)$ &0.756& 0.731 & 0.765 &  & 0.756 & 0.778 &  & 0.739 & 0.798 &  & 0.825 & 1.000 &  & 0.721 & 0.863 & 1.000 \\
 &(0.087)& (0.001) & (0.008) &  & (0.006) & (0.006) &  & (0.008) & (0.010) &  & (0.007) & (0.000) &  & (0.008) & (0.006) & (0.000) \\
$\Pr(X_3=1|0)$ &0.098& 0.076 & 0.130 &  & 0.111 & 0.029 &  & 0.060 & 0.136 &  & 0.099 & 0.079 &  & 0.089 & 0.105 & 0.105 \\
 &(0.029)& (0.030) & (0.028) &  & (0.028) & (0.029) &  & (0.030) & (0.028) &  & (0.029) & (0.029) &  & (0.028) & (0.027) & (0.028) \\
 \hline
$p$-value ($q_1=0$) &0.000 &0.000 & 0.000 &  & 0.000 & 0.014 &  & 0.000 & 0.000 &  & 0.000 & 0.006 &  & 0.000 & 0.000 & 0.500 \\
$p$-value($q_0=0$) &0.182  &0.129 & 0.100 &  & 0.129 & 0.153 &  & 0.118 & 0.151 &  & 0.080 & 0.182 &  & 0.073 & 0.266 & 0.445 \\
 sample size &1270&740 & 525&& 1022&81&& 463&535&&578&194&&840&351&79\\
 \hline\hline
\end{tabular}
\begin{tablenotes}
\item Note: $X^*=1$ stands for non-heterosexuality. $\Pr(X_j=1|1)$ and $\Pr(X_j=1|0)$ are $\Pr(X_j=1|X^*=1)$ and $\Pr(X_j=1|X^*=0)$, respectively, for $j=1,2,3.$  $X_1=1, X_2=1$, and $X_3=1$ represent affirmative answers to  ``heterosexuality", ``same-sex attraction", and ``same-sex sexual experience", respectively.
\item The column ``overall" includes unconditional estimates, all other columns of estimates are conditional on demographics.
\item The $p$-value ($q_1=0$) and ($q_0=0$) are for the hypotheses $q_1=0$ v.s $q_1>0$, and $q_0=0$ v.s $q_0>0$, respectively.
\item The results are estimated using the extreme estimator proposed in Section \ref{section: estimation methods}. Standard errors are bootstrapped 1000 times.

\end{tablenotes}
\end{threeparttable}
}
\end{center}
\end{table}

There are several main findings from the estimates of sexual orientation. First, we provide novel results on respondents' misreporting behaviors.
 When the respondents are asked directly about sexual orientation (in the first sensitive question), the estimate of $\Pr(X_1=1|X^*)$ and $\Pr(X_1=1|X^*,Z)$ measure respondents' misreporting and its dependence on demographics. Table \ref{table: estimation results of LGBT probability} shows that a substantial portion (about 28.8\%) of non-heterosexual respondents report themselves to be heterosexual.  The misreporting  depends heavily on demographics. Male, Black, Christian, and Republican respondents report a much lower proportions of non-heterosexuality than their counterpart groups. For example,
47.3\% of non-heterosexual Republicans claim themselves to be heterosexual, while the percentage is only 25.6\% for Democrats, which is 45.9\% smaller. The misreporting among heterosexual respondents is much smaller. Only 2.9\% of heterosexual respondents claim non-heterosexuality, and that estimate
is not significantly different from zero for all of the demographic groups. While misreporting in sensitive survey questions has been documented as a serious issue, the quantitative results on misreporting behaviors are largely missing in the literature. Our estimates quantify the misreporting behaviors and thus improve our understanding of respondents' reporting strategies across different demographic groups.

The estimates of misreporting also can be used to evaluate the applicability of the modified LE. As discussed in Section \ref{section: modified LE}, the applications of the modified LE rely on
the probabilities $q_1$ and $q_0$: the proportion of heterosexual and homosexual respondents who misreport their sexuality, respectively. In our notation, $q_1\equiv \Pr(X_1=1|X^*=1)$ and $q_0\equiv 1-\Pr(X_1=1|X^*=0)$, and truthful reporting means $q_1=0$ and $q_0=0$.
Based on our estimates, $q_1=0$ is rejected while $q_0=0$ is not rejected at the 5\% significance level.\footnote{The only exception is that we fail to reject $q_1=0$ for respondents older than 50. This may be due to the small sample size of this group.} Our findings suggest that heterosexual respondents report their sexuality truthfully while non-heterosexual respondents significantly  misreport. This provides new insights regarding misreporting on sensitive questions, because reporting strategies of the two groups have not been identified separately in the literature.

Second, we directly estimate the proportion of non-heterosexual respondents and their dependence on demographics. We find that 11.3\% of respondents are non-heterosexual. That proportion varies significantly across demographic groups, especially across gender, religion, politics, and age. The largest difference is observed between Democrats and Republicans. The proportions of non-heterosexuality are 14.9\% among Democrats and 4.7\% among Republicans. Our estimate of the proportion of non-heterosexuality is in the upper tail of the existing results,\footnote{In a review article, \cite{gonsiorek1995definition} suggest that the current prevalence of predominant same-sex orientation is 4-17\%.} and this may be due to the sample being younger, more liberal, and less religious than general population.
 As pointed out in \cite{gonsiorek1995definition}, the variation in the existing results is caused by different definitions and measures of LGBT across surveys. Our quantitative results reveal heterogenous responses to the survey questions among different sexual orientations and uncover the distribution of non-heterosexuality across demographic groups. These shed light on the importance of survey design and the choice of demographics of respondents. For example, we find a larger proportion of non-heterosexual Democrats than Republican among respondents. One implication is thus that respondents of different political views should be evenly recruited for a survey in order to obtain reasonable results.


Finally, it is worth noting that our estimates are not directly comparable to those from LE in \cite{coffman2017size} because the assumptions of our method and LE are different. Nevertheless,
our conclusions are consistent with the general findings in \cite{coffman2017size}: that respondents' non-heterosexual identity and same-sex sexual experience are substantially under-reported, but
not same-sex attraction, based on the modified LE approach. We estimate that only 71.2\% and 69.5\% of homosexual respondents
report to be non-heterosexual and to have had same-sex sexual experiences, respectively. However, almost all of them claim that they are sexually attracted to people of the same sex.

\subsection{LGBT-related sentiment}
In this section, we analyze LGBT-related sentiment and its dependence on the demographics of respondents.
In this application, $X^*=1$ indicates negative sentiment toward the LGBT population. The three measurements of that sentiment are the answers to the following questions,\footnote{There are five questions on LGBT-related sentiment in \cite{coffman2017size}. We choose the first three of them for our analysis.}
 \begin{enumerate}
 \setlength\itemsep{0.001cm}
\item Do you think marriages between gay and lesbian couples should be recognized by the law as valid, with the same rights
as heterosexual marriages?
\item Would you be happy to have an openly lesbian, gay, or bisexual manager at work?
\item Do you believe it should be
illegal to discriminate in hiring
based on someone's sexual orientation?
\end{enumerate}
The choice of covariates $Z$ is the same as in our sexual orientation analyses. The identification strategy requires that the responses are independent, given a respondent's sentiment and demographics. Respondents with negative sentiment would respond to ``support same-sex marriage", ``happy with LGB manager", and ``believe that it is illegal to discriminate LGBT people" differently from those with positive sentiment. Moreover, we assume that respondents with positive sentiment are more likely to support same-sex marriage than their counterparts, given their demographics, i.e., $\Pr(X_1=1|X^*=0,z)>\Pr(X_1=1|X^*=1,z)$. The results presented in Table \ref{table: estimation results of LGBT sentiment} are based on the same estimating procedure as in Section \ref{subsection: estimating sexual orientation}. The  probability of negative sentiment conditional on demographics and their 95\% confidence intervals also are presented in the right panel of Figure \ref{fig:effects of demographics on sexuality and sentiment}.

We summarize the main findings as follows. First, 13.3\% of the respondents are estimated to have a negative sentiment toward the LGBT population. If we assume that all of these respondents are heterosexual, then about 15\% (calculated from 13.3\%/(1-11.3\%), where 11.3\% is the estimated size of the LGBT population) of heterosexual respondents are not friendly to the LGBT-population. More importantly, the respondents with negative sentiment are significantly divided. Black, Christian, and Republican respondents are substantially more negative toward the LGBT population than their counterparts. For example, 22\% of Christian respondents have negative sentiments while the percentage is only 3.2\% for non-religious respondents.

Second, the responses of those respondents with negative sentiment to the three questions display significant divergence.
Having an LGB manager at work is the least accepted among those respondents. Only 2.0\% of them answered affirmatively. By contrast, 58.5\% of them agree that discrimination against LGBT people is illegal. The supporting rate  for same-sex marriage is 17.8\%.  These estimates exhibit sharp divides among demographic groups and the attitude toward same-sex marriage is the most divergent among the three questions. Specifically, the proportion of white respondents supporting same-sex marriage is 5.3 times the size of black respondents; that of nonreligious respondents is 2.9 times the size of Christian respondents, and that of Democrats is 5.0 times the size of Republicans.  The response to ``happy with LGB manager" question is significantly influenced only by race. The divergence on ``believe it is illegal to discriminate LGBT people" is most obvious across gender, religion, and politics. These observations imply that the magnitude of antigay sentiment differs significantly across its three dimensions (or measurements) and demographic groups.

On the other hand, among respondents with positive sentiments, 96.4\% and 90.9\% of them are okay with an LGB manager at work and supporting same-sex marriage, respectively. The proportions are not significantly different from 100\% at the 5\% significance level. A slightly smaller proportion (89.8\%) of them believe that discrimination against LGBT people is illegal. The three estimates are not statistically different at the 5\% significance level.

Our analysis of anti-gay sentiment for those who have negative sentiments is novel, and the findings greatly enhance our understanding of the issue. Neither LE nor modified LE allows us to derive such quantitative evidence as those models can only provide weighted average results of the two groups with opposite sentiments.
\begin{table}
\centering
\caption{Results of estimation: LGBT-related sentiment}
\vspace{-0.6cm}
\label{table: estimation results of LGBT sentiment}
\begin{center}
\scalebox{0.70}{\begin{threeparttable}
\begin{tabular}{llllllllllllllllll}
\hline\hline
\cline{3-17}\\
& &\multicolumn{2}{c}{gender} &&\multicolumn{2}{c}{race} &&\multicolumn{2}{c}{religion}&&\multicolumn{2}{c}{politics}&&\multicolumn{3}{c}{age} \\
\cline{3-4} \cline{6-7}\cline{9-10}\cline{12-13} \cline{15-17}
\\
parameter&overall&male & female && white & black && christian& no reli.&& dec.& rep.&&$<31$ &31-50&$>50$ \\
\hline$\Pr(X^*=1)$ &0.133 & 0.136 & 0.135 &  & 0.115 & 0.383 &  & 0.221 & 0.032 &  & 0.115 & 0.321 &  & 0.125 & 0.160 & 0.329 \\
 &(0.076) & (0.081) & (0.319) &  & (0.261) & (0.067) &  & (0.021) & (0.214) &  & (0.116) & (0.036) &  & (0.075) & (0.140) & (0.063) \\
 \hline
$\Pr(X_1=1|1)$ &0.178 &0.160 & 0.213 &  & 0.154 & 0.029 &  & 0.092 & 0.265 &  & 0.403 & 0.081 &  & 0.263 & 0.189 & 0.192 \\
 &(0.001)& (0.005) & (0.055) &  & (0.032) & (0.014) &  & (0.001) & (0.195) &  & (0.124) & (0.001) &  & (0.032) & (0.045) & (0.001) \\
$\Pr(X_1=1|0)$ &0.909& 0.920 & 0.896 &  & 0.913 & 0.943 &  & 0.817 & 0.987 &  & 0.968 & 0.728 &  & 0.934 & 0.874 & 0.821 \\
 &(0.137)& (0.148) & (0.499) &  & (0.378) & (0.081) &  & (0.001) & (0.233) &  & (0.153) & (0.001) &  & (0.104) & (0.274) & (0.062) \\
 \hline
$\Pr(X_2=1|1)$ &0.020& 0.017 & 0.032 &  & 0.000 & 0.314 &  & 0.017 & 0.000 &  & 0.144 & 0.087 &  & 0.095 & 0.000 & 0.000 \\
 & (0.096)&(0.139) & (0.185) &  & (0.169) & (0.082) &  & (0.001) & (0.175) &  & (0.128) & (0.000) &  & (0.062) & (0.187) & (0.056) \\
$\Pr(X_2=1|0)$ &0.964& 0.944 & 0.996 &  & 0.970 & 0.886 &  & 0.930 & 0.960 &  & 0.994 & 0.923 &  & 0.961 & 1.000 & 1.000 \\
 &0.088& (0.046) & (0.402) &  & (0.293) & (0.015) &  & (0.001) & (0.223) &  & (0.157) & (0.001) &  & (0.084) & (0.151) & (0.037) \\
 \hline
$\Pr(X_3=1|1)$ &0.585& 0.534 & 0.666 &  & 0.581 & 0.615 &  & 0.547 & 0.295 &  & 0.705 & 0.527 &  & 0.617 & 0.625 & 0.500 \\
 &(0.004)& (0.004) & (0.003) &  & (0.004) & (0.005) &  & (0.005) & (0.002) &  & (0.003) & (0.005) &  & (0.003) & (0.004) & (0.006) \\
$\Pr(X_3=1|0)$ &0.898& 0.890 & 0.907 &  & 0.904 & 0.859 &  & 0.865 & 0.913 &  & 0.919 & 0.889 &  & 0.911 & 0.885 & 0.830 \\
 &(0.008)& (0.010) & (0.029) &  & (0.009) & (0.016) &  & (0.008) & (0.016) &  & (0.008) & (0.010) &  & (0.006) & (0.009) & (0.009) \\
  sample size &1270&740 & 525&& 1022&81&& 463&535&&578&194&&840&351&79\\
 \hline\hline
\end{tabular}
\begin{tablenotes}
\item Note: $X^*=1$ stands for \textit{negative} sentiment toward the LGBT population. $\Pr(X_j=1|1)$ and $\Pr(X_j=1|0)$ are $\Pr(X_j=1|X^*=1)$ and $\Pr(X_j=1|X^*=0)$, respectively, for $j=1,2,3.$ $X_1=1, X_2=1$, and $X_3=1$ represent affirmative answers to  ``support same-sex marriage", ``happy with LGB manager", and ``illegal to discriminate", respectively.
\item The results are estimated using the extreme estimator proposed in Section \ref{section: estimation methods}. Standard errors are bootstrapped 1000 times.
\end{tablenotes}
\end{threeparttable}
}
\end{center}
\end{table}

\begin{figure}[tbp]
\caption{Effects of demographics on homosexuality and antigay sentiment}
\label{fig:effects of demographics on sexuality and sentiment}\centering
\begin{subfigure}[t]{0.5\textwidth}
        \includegraphics[height=2.4in]{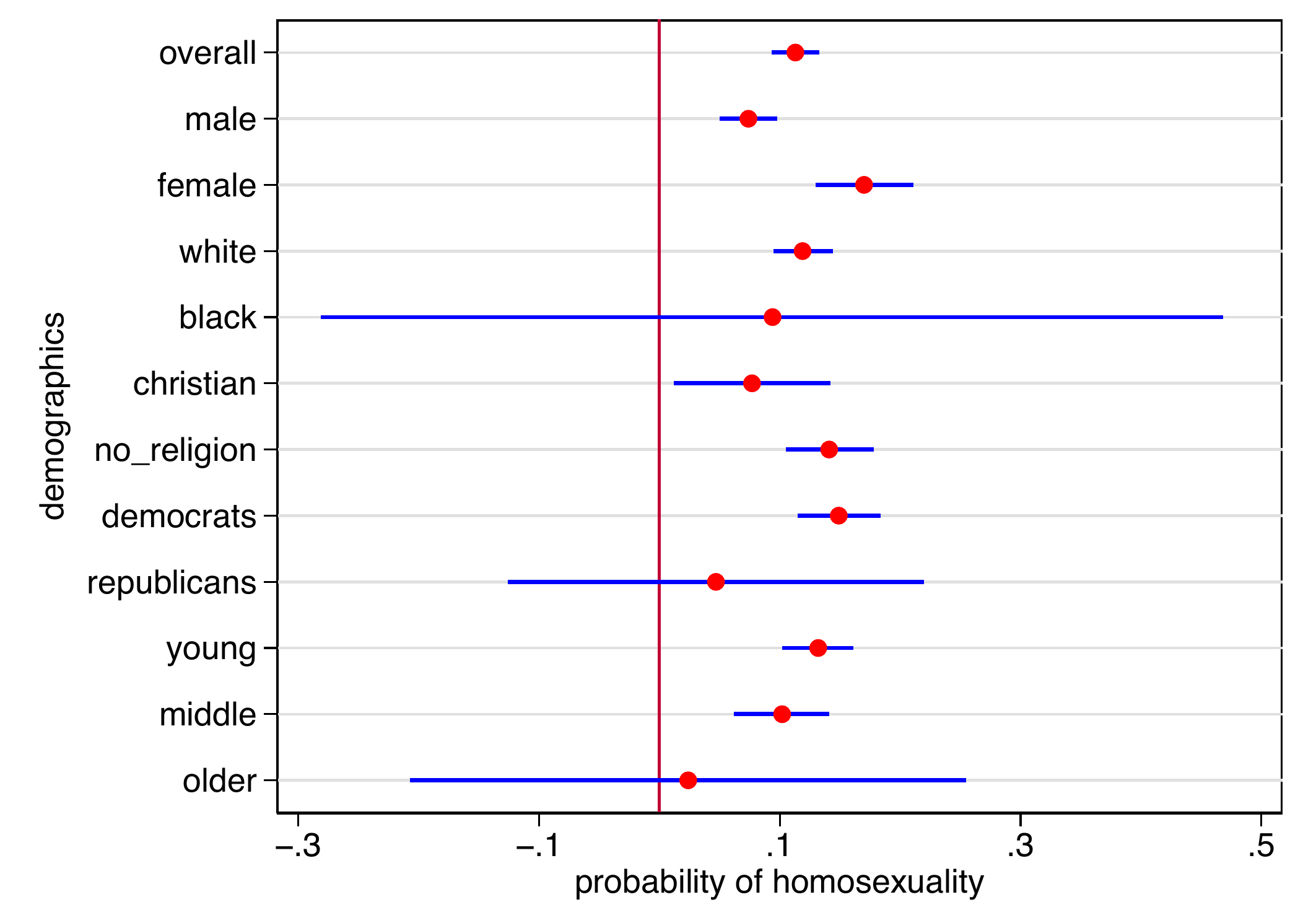}
        \caption{\scriptsize{homosexuality}}
    \end{subfigure}
\begin{subfigure}[t]{0.49\textwidth}
        \includegraphics[height=2.4in]{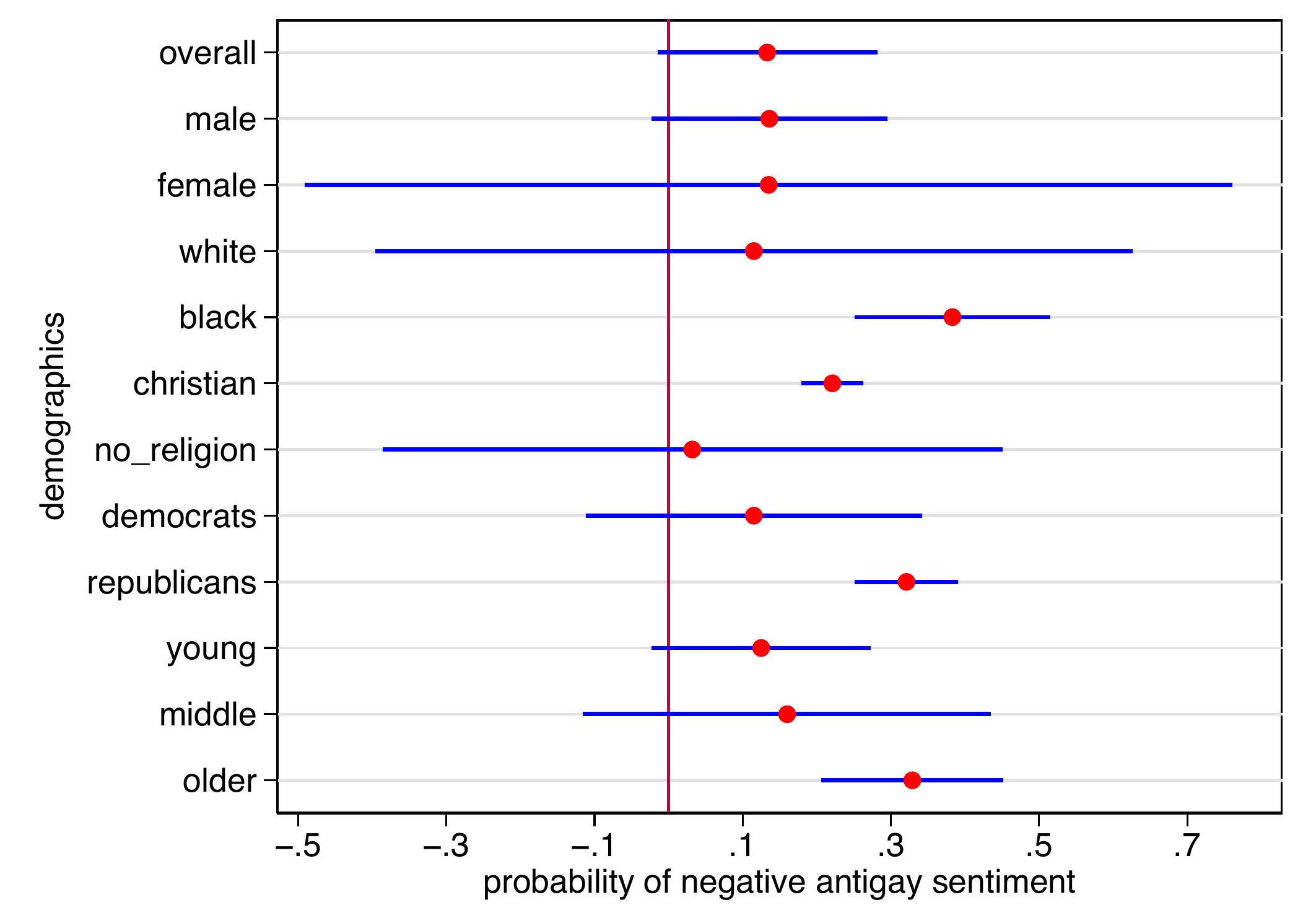}
       \caption{\scriptsize{antigay sentiment}}
    \end{subfigure}
\par
\vspace{0.3cm}
\par
\begin{minipage}{0.7\textwidth} 
{\footnotesize Note: The red dots and the blue lines represent point estimates and their 95\% confidence intervals, respectively.
\par}
\end{minipage}
\end{figure}

\section{Conclusions}
This article studies the eliciting of information from sensitive survey questions. We make two main points in the paper. First,
it is necessary to test the assumptions of the widely used LE before applying the method to obtain estimates about sensitive information. We prove that the assumptions of LE can be tested rigorously. We find that they are violated in the majority of empirical studies. That violation implies invalidity of LE and problematic conclusions based on using the method. Second,
information can be elicited from sensitive survey questions by applying our proposed technique. To implement it, one would ask all of the respondents to answer three or more survey questions related to the information to be elicited. Random assignment of groups is not necessary. Our technique also can address
measurement error in arbitrary forms, and it allows us to recover information at a disaggregated level. Applying our technique to survey data on the LGBT-population and LGBT-related sentiments leads to several novel findings that are obscured by using LE or other existing survey methods.

There are several avenues suggested for future research. One could apply our technique to other experimental methods with measurement error or unobserved heterogeneity: for example, the Goldberg paradigm experiments that are used to measure discrimination (see details in \cite{bertrand2017field}). Another possibility is to investigate how to optimally design survey questions and to collect respondents' demographics in order to obtain the best estimate from survey responses.

\bibliographystyle{aer}
\bibliography{new}

\bigskip

\bigskip

\bigskip

\bigskip

\noindent\textbf{{\LARGE {Appendix}}} \appendix

\numberwithin{equation}{section} \renewcommand\thesection{\Alph{section}}

\section{Proofs\label{appendix_section:proof}}

\paragraph{Proof of proposition \ref{proposition: model restrictions under randomization}.}
Under Assumption \ref{asp: random assignment}, $\Pr(R_0=j)=\Pr(R_1=j)\equiv P(j)$ is the probability of outcome $j\in\mathcal{J}_0$ under respondents' true preference. By Assumption \ref{asp: misreporting probability},  the probability of outcome $j\in\mathcal{J}_0$ in the control group is \begin{equation}
P_0(j)=(1-p_0){P(j)}+\frac{p_0}{J+1},
\end{equation}
or
\begin{equation}
P(j)=\frac{P_0(j)}{1-p_0}-\frac{p_0}{(J+1)(1-p_0)}.
\end{equation}
Consider the random variable $U\equiv R_1+X^*$, where $R_1$ and $X^*$ represent the outcomes for the nonsensitive and sensitive questions under respondents' true preference, respectively. According to Assumption \ref{asp: misreporting probability},  the probability of an outcome $j\in\mathcal{J}_1$ in treatment group can be obtained by convolution of probability distributions:
\begin{align}P_1(j)&=(1-p_0)\left(\sum_{j-r_1\in\{0,1\}}P(R_1=r_1)P(X^*=j-r_1)\right)+\frac{p_1}{J+2}.
\end{align}
When $j=0$,
\begin{align}
P_1(0)=(1-p_1)P(0)P(X^*=0)+\frac{p_1}{J+2}=\frac{1-p_1}{1-p_0}\bigg((1-\delta)P_0(0)-\frac{(1-\delta)p_0}{J+1}\bigg)+\frac{p_1}{J+2},
\end{align}
where $\delta=P(X^*=1)$. When $1\leq j<J+1$,
\begin{eqnarray}
P_1(j)&=&(1-p_1)\left(P(j)(1-\delta)+P(j-1)\delta\right)+\frac{p_1}{J+2}\nonumber\\
&=&(1-p_1)\left[\left(\frac{P_0(j)}{1-p_0}-\frac{p_0}{(J+1)(1-p_0)}\right)(1-\delta)\right.\nonumber\\
&+&\left.\left(\frac{P_0(j-1)}{1-p_0}-\frac{p_0}{(J+1)(1-p_0)}\right)\delta\right]+\frac{p_1}{J+2}\nonumber\\
&=&\frac{1-p_1}{1-p_0}\left(\delta P_0(j-1)+(1-\delta)P_0(j)-\frac{p_0}{J+1}\right)+\frac{p_1}{J+2}.
\end{eqnarray}
When $j=J+1$,\begin{align}
P_1(J+1)=(1-p_1)P(J)P(W=1)+\frac{p_1}{J+2}=\frac{1-p_1}{1-p_0}\bigg(\delta P_0(J)-\frac{\delta p_0}{J+1}\bigg)+\frac{p_1}{J+2}.
\end{align}
\paragraph{Proof of identification of $\delta, p_0$ and $p_1$ from equation (\ref{equation: mean difference of responses with misreporting}).} Without loss of generality, we assume $J=3$. The cases with $J>3$ can be proved analogously.

When $J=3$, all the model restrictions are
\begin{eqnarray}\label{eq: conditions for J=3}
P_1(0)&=&\frac{1-p_1}{1-p_0}\bigg((1-\delta)P_0(0)-\frac{(1-\delta)p_0}{4}\bigg)+\frac{p_1}{5},\nonumber\\
P_1(j)&=&\frac{1-p_1}{1-p_0}\left(\delta P_0(j-1)+(1-\delta)P_0(j)-\frac{p_0}{4}\right)+\frac{p_1}{5}, j=1,2,3,\nonumber\\
P_1(4)&=&\frac{1-p_1}{1-p_0}\bigg(\delta P_0(3)-\frac{\delta p_0}{4}\bigg)+\frac{p_1}{5},
\end{eqnarray}
where $0\leq P_t(j)\leq 1$,\hspace{0.2cm} $\sum_{j=0}^{J+t} P_t(j)=1$,  and $t=0, 1$.

Note that the five equations above are nonlinear in $\delta$, $p_0$ and $p_1$. Below we discuss sufficient but not necessary conditions under which Equation (\ref{eq: conditions for J=3}) sustains a solution. Since the equation is linear in $\delta$ when $j=1, 2, 3$, we use the equations for $j=1, 2, 3$ to get
\begin{eqnarray}
\frac{P_1(3)-P_1(2)}{P_1(2)-P_1(1)}=\frac{(1-\delta)P_0(3)+(2\delta-1)P_0(2)-\delta P_0(1)}{(1-\delta)P_0(2)+(2\delta-1)P_0(1)-\delta P_1(0)},
\end{eqnarray}
which identifies $\delta$ whenever $(P_1(3)-P_1(2))\big(2P_0(1)-P_0(2)-P_1(0)\big)\neq (P_1(2)-P_1(1))\big(2P_0(2)-P_0(3)-P_1(1)\big)$. We then combine the first and last equations ($j=0, 4$) and two other consecutive equations (e.g., $j=1,2$ or $j=2,3$) to obtain
\begin{eqnarray}
\frac{P_1(4)-P_1(0)}{P_1(2)-P_1(1)}=\frac{\delta P_0(3)-(1-\delta)P_0(0)+(1-2\delta)p_0/4}{(1-\delta)P_0(2)+(2\delta-1)P_0(1)-\delta P_1(0)}.
\end{eqnarray}
The parameter $p_0$ is identified if $\delta\neq 1/2$. For a given identified pair $(\delta, p_0)$, any condition in Equation (\ref{eq: conditions for J=3}) identifies $p_1$. The identified parameters under the sufficient conditions we provided are not necessarily in $[0,1]$. Therefore, Equation (\ref{eq: conditions for J=3}) sustains at most one solution.
Our argument above can be generalized to $J>3$.

\paragraph{Proof of Corollary 1.}
Starting from equation (\ref{equation: relationship between observed responses with misreporting}), we first calculate $\mathbb{E}(Y_1)$.
\begin{eqnarray}
\mathbb{E}(Y_1) &=& \sum_{j=0}^{J+1}P_1(j)j\nonumber\\
                &=&0*\left\{ \frac{1-p_1}{1-p_0}\bigg((1-\delta)P_0(0)-\frac{(1-\delta)p_0}{J+1}\bigg)+\frac{p_1}{J+2}\right\}\nonumber\\
                &&+1*\left\{\frac{1-p_1}{1-p_0}\left(\delta P_0(0)+(1-\delta)P_0(1)-\frac{p_0}{J+1}\right)+\frac{p_1}{J+2}\right\}\nonumber\\
                &&\vdots\nonumber\\
                &&+J*\left\{\frac{1-p_1}{1-p_0}\left(\delta P_0(J-1)+(1-\delta)P_0(J)-\frac{p_0}{J+1}\right)+\frac{p_1}{J+2}\right\}\nonumber\\
                &&+(J+1)*\left\{\frac{1-p_1}{1-p_0}\bigg(\delta P_0(J)-\frac{\delta p_0}{J+1}\bigg)+\frac{p_1}{J+2}\right\}\nonumber\\
                &=&\frac{1-p_1}{1-p_0}\left\{\delta P_0(0)+(1+\delta)P_0(1)+\cdots (J+\delta)P_0(J)-(1+2+\cdots+J)\frac{p_0}{J+1}+\delta p_0\right\}\nonumber\\
                &&+(1+2+\cdots+J+1)\frac{p_1}{J+2}\nonumber\\
                &=&\frac{1-p_1}{1-p_0}\left\{\sum_{j=0}^{J}P_0(j)j+\delta\sum_{j=0}^{J}P_0(j)-\frac{Jp_0}{2}-\delta p_0\right\}+\frac{(J+1)p_1}{2}\nonumber\\
                &=&\frac{1-p_1}{1-p_0}\left\{\mathbb{E}(Y_0)+\delta(1-p_0)-\frac{Jp_0}{2}\right\}+\frac{(J+1)p_1}{2}.
\end{eqnarray}
Based on the expression above, we have
\begin{eqnarray}
\mathbb{E}(Y_1) -\mathbb{E}(Y_0)
                &=&\frac{1-p_1}{1-p_0}\left\{\mathbb{E}(Y_0)-\frac{1-p_0}{1-p_1}\mathbb{E}(Y_0)+
                \delta(1-p_0)-\frac{Jp_0}{2}\right\}+\frac{(J+1)p_1}{2}\nonumber\\
                &=&\delta(1-p_1)+\frac{p_0-p_1}{1-p_0}\mathbb{E}(Y_0)-\frac{J(1-p_1)p_0}{2(1-p_0)}+\frac{J+1}{2}p_1.
\end{eqnarray}

\section{An Extension: Strategic Misreporting\label{appendix: an alternative misreporting strategy}}
In this section, we discuss an alternative misreporting strategy and its implications.
We assume that all the respondents without sensitive information answer survey questions truthfully regardless of their group assignment. Respondents with sensitive information misreport with probability $p$ only if truthful reporting would cause privacy disclosure (in our case, respondents may cause privacy disclosure when $R_1=J$). When a respondent with sensitive information chooses to misreport, we assume that she answers no to the sensitive question.
Following the notation in Proposition \ref{proposition: model restrictions under randomization}, the assumptions of strategic misreporting can be summarized as:
\begin{eqnarray}\Pr(X^*=1|R_1<J)=\delta, \hspace{0.2cm}\Pr(X^*=1|R_1=J)=\delta(1-p).
\end{eqnarray}
The assumption that all the respondents without sensitive information report truthfully implies that
\begin{eqnarray}\Pr(R_0=j)=\Pr(R_1=j)= P(j), \hspace{0.2cm} j\in\mathcal{J}_0.
\end{eqnarray}
Under the assumption above, only the outcomes $J$ and $J+1$ would be affected by misreporting. For $j=0, 1, \cdots, j-1$, no misreporting is involved, we have
\begin{eqnarray}P_1(j)&=&\sum_{j-r_1\in\{0,1\}}\Pr(R_1=r_1)\Pr(X^*=j-r_1)\nonumber\\
                                      &=&\begin{cases}
(1-\delta)P_0(0),& j=0. \\
(1-\delta)P_0(j)+\delta P_0(j-1)& j=1, 2, \cdots, j-1.
\end{cases}
\label{equation: alternative strategy part 1}
\end{eqnarray}
When $j=J$, the outcome is determined by both truthful and mis-reporting because those who misreport choose to report $j=J$,
    \begin{eqnarray}\label{equation: alternative strategy part 2}
       P_1(J)=(1-\delta)P_0(J)+\delta P(J-1)+ \delta p P_0(J).
    \end{eqnarray}
The outcome $j=J+1$ is observed only when respondents with sensitive information answer yes to all the $J$ nonsensitive questions under their true preference and decide to reveal their sensitive information even though doing so would discloses their privacy. Therefore, we have
    \begin{eqnarray}\label{equation: alternative strategy part 3}
        P_1(J+1)=\delta (1-p)P_0(J).
    \end{eqnarray}
Equations (\ref{equation: alternative strategy part 1})-(\ref{equation: alternative strategy part 3}) summarize all the model restrictions under the strategic misreporting. The parameters $\delta$ and $p$ are both identified given $J\geq 1$. The proof is similar to that of equation (\ref{equation: mean difference of responses with misreporting}), thus is omitted.

 Based on the model restrictions (\ref{equation: alternative strategy part 1})-(\ref{equation: alternative strategy part 3}), we can test LE with this misreporting strategy by using a testing procedure that is similar to Section \ref{subsection: test misreporting behavior}. The testing results of the five articles as in Section \ref{section: application of the test} are presented in Table \ref{table: testing the assumption of LE alt}. The findings from the test are similar to that in Section \ref{section: application of the test}. When $p\neq 0$, we reject the null hypothesis for 12 and 13 out of 19 questions at the 5\% and 10\% significance levels, respectively, accounting for 63\% and 68\%. By design, the testing results for $p=0$ is the same as that for $p_1=p_2=0$ in Section \ref{section: application of the test}. The model is rejected for 53\% (10 out of 19) and 63\% (12 out of 19) of questions
at the 5\% and 10\% significance level, respectively, in both specifications.

\begin{table}
\centering
\caption{Results of estimation and testing, strategic misreporting}
\vspace{-0.8cm}
\label{table: testing the assumption of LE alt}
\begin{center}
\scalebox{0.65}{\begin{threeparttable}
\begin{tabular}{llllllll}
\hline\hline
  & & &\multicolumn{2}{c}{GMM estimate} &&\multicolumn{2}{c}{testing results} \\
\cline{4-5} \cline{7-8}\vspace{-0.4cm}
\\
sensitive question&\tabincell{c}{sample\\ size}  & \tabincell{c}{mean \\ difference}& $\delta$& $p$  &&\tabincell{c} {$p\neq 0$}&\tabincell{c} {$p=0$}\\
\hline
favorable view of CCP &1576 & 0.057 & --- & --- && \xmark & \xmark\\
& & (0.057) & & &  \\
consider self Hong Kongese&1576 & 0.816 &---&---
&&    \xmark&\xmark\\
& & (0.048)& \\
support for HK independence &1576 & 0.521 & 0.521 & 0.209  && $\checkmark$ &\xmark\\
& & (0.052) & (0.052) & (0.303) &  \\
\tabincell{l}{support violence in pursuit of\\  HK's political rights} &1576 & 0.389 & 0.423 & --- && \xmark &$\checkmark$ \\
& & (0.048) & (0.052) & &  \\
\hline
\tabincell{l}{I completely trust the \\central government of China}&1807 & 0.290 &---&---
&& \xmark & \xmark \\
& & (0.038)  \\
 \hline
\tabincell{l}{do you consider yourself to be heterosexual?}&2516 & 0.894 &---&---
&& \xmark & $\dagger$\\
& & (0.036) & & &  \\
\tabincell{l}{are you sexually attracted to members of \\the same sex?}&2516 & 0.258 &---&---
&& $\dagger$& $\dagger$\\
& & (0.035)  \\
\tabincell{l}{ have you had a sexual experience with \\ someone of the same sex?}&2516
& 0.545 &---&--- && \xmark  &\xmark \\
& & (0.040)  \\
\tabincell{l}{do you think marriages between gay and \\ lesbian couples should be recognized by \\ the  law as valid, with the same rights \\as heterosexual marriages?}&2516 & 0.986&---&---
 && \xmark & \xmark \\
& & (0.032)  \\
\tabincell{l}{ would you be happy to have an openly \\lesbian, gay, or bisexual manager at work?}&2516 & 0.912 &---&---
&& \xmark & \xmark \\
& & (0.038)  \\
\tabincell{l}{do you believe it should be \\ illegal to discriminate  in hiring\\ based on someone's sexual orientation?}&2516 & 0.806
&---&---
 && \xmark &\xmark  \\
& & (0.031)  \\
\tabincell{l}{do you believe lesbians and gay men \\ should be allowed to adopt children?}&2516 & 0.878 & --- & ---  && \xmark & \xmark \\
& & (0.034)  \\
\tabincell{l}{do you think someone who is \\ homosexual can change their sexual\\ orientation if they choose to do so?}&2516 & 0.186 & 0.187 & --- && $\checkmark$ &$\checkmark$ \\
 && (0.034) & (0.031) &  &  \\
\hline
\tabincell{l}{w/o smartcards system, members\\ of this household have been asked\\ by officials to lie about the\\ amount of work they did on NREGS} &917 & 0.044 &0.014 & --- && $\checkmark$&$\checkmark$\\
 && (0.059) & (0.072) & & \\
 \tabincell{l}{w/o smartcards system, members\\ of this household have been given\\ the chance to meet with the CM of AP\\ to discuss problems with NREGS?}&897 & 0.174 & 0.133 & 0.275 && $\checkmark$  &\xmark \\
& & (0.062) & (0.088) & (0.233) &  \\
\tabincell{l}{w/ smartcards system, members\\ of this household have been asked\\ by officials to lie about the\\ amount of work they did on NREGS}&2300 & 0.106 &---&---
&&\xmark &\xmark \\
& & (0.036)  \\
\tabincell{l}{w/ smartcards system, members\\ of this household have been given\\ the chance to meet with the CM of AP\\ to discuss problems with NREGS?}&2276 & 0.154 & 0.140 & ---  && $\checkmark$ & $\checkmark$  \\
 && (0.037) & (0.053) &  &  \\
\hline
\tabincell{l}{treated voters differently by religion/caste} &3833&	0.242&	0.248&---&&	
$\checkmark$  &$\checkmark$ \\
&& (0.027) & (0.029) &  &  \\
\tabincell{l}{attempted to influence voting}&3850	&0.142	&---&--- &&\xmark &\xmark \\
& & (0.038)  \\
 \hline\hline
\end{tabular}
\begin{tablenotes}
\item Note: The five panels (from top to bottom) are results for \cite{cantoni2019protests}, \cite{chen2019impact}, \cite{coffman2017size}, \cite{muralidharan2016building}, and \cite{neggers2018enfranchising}, respectively.
\item \xmark, $\dagger$, and $\checkmark$ indicate $p<0.05$, $p<0.1$, and $p>0.1$, respectively. Standard errors in parentheses are bootstrapped $1000$ times. When only $\hat{\delta}$ is provided, then $p=0$; if both $\hat{\delta}$ and $\hat{p}$ are provided, then $p\ne 0$.
\end{tablenotes}
\end{threeparttable}
}
\end{center}
\end{table}

\section{Estimation\label{section: estimation methods}}
In this section, we propose three methods to estimate the model primitives identified in Theorem \ref{theorem: results of nonparametric identification}. The first two methods are nonparametric and follow the identification procedure closely. The third one utilizes maximum likelihood estimation (MLE) to estimate the model primitives.  The data sample is represented by $\{X_{ji}, Z_i\}, j=1, 2, 3, i=1, 2, \cdots, n$, based on the $n$ respondents' answers to three questions ($X_{ji}$) and their characteristics ($Z_i$).
\subsection{Discrete covariates}
When the vector of covariates $Z$ takes only a few discrete values, such as gender, political affiliation, and race, it is convenient to estimate the model using a nonparametric approach based on the identification procedure. According to the main identification equation (\ref{equation: main identification equation}),  the eigenvector matrix $M_{X_1|X^*,z}$ can be written as a function of the observed matrix $M_{X_1, x_2,X_3|z}M^{-1}_{X_1, X_3|z}$,
\begin{eqnarray*}
M_{X_1|X^*,z}=\psi(M_{X_1, x_2,X_3|z}M^{-1}_{X_1, X_3|z}),
\end{eqnarray*}
where $\psi(\cdot)$ is a function determined by eigenvalue-eigenvector decomposition of a matrix.\footnote{Although the explicit expression of $\psi(\cdot)$ is complicated, a general
result in \cite{andrew1993derivatives} shows that it is analytic. The analytical properties of $\psi(\cdot)$ guarantees the regular asymptotic properties of the nonparametric estimators in this section.}
Thus, we estimate $M_{X_1|X^*,z}$ as
\begin{eqnarray}\label{equation: nonparametric estimation of eigenvector matrix}
\widehat{M}_{X_1|X^*,z}=\psi(\widehat{M}_{X_1, x_2,X_3|z}\widehat{M}^{-1}_{X_1, X_3|z}),
\end{eqnarray}
where $\widehat{M}_{X_1, x_2,X_3|z}$ and $\widehat{M}^{-1}_{X_1, X_3|z}$ can be estimated by
\small
\begin{eqnarray}
\left(\widehat{M}_{X_1, x_2,X_3|z}\right)_{l,k}=\sum\nolimits_{i=1}^{n}  \mathbbm{1}\left(X_{1i}=l-1, X_{2i}=x_2, X_{3i}=k-1\right)\mathbbm{1}\left(Z_i=z\right)/\sum\nolimits_{i=1}^{n}\mathbbm{1}\left(Z_i=z\right)
\end{eqnarray}
\normalsize
for $k, l=1, 2$. The matrix $\widehat{M}_{X_1, X_3|z}$ is estimated analogously.
Once $\widehat{M}_{X_1|X^*,z}$ is obtained, we follow equation (\ref{equation: estimation of the conditional probabilities of sensitive information}) to estimate $\Pr(X^*|z)\equiv [\Pr(X^*=0|z) \hspace{0.3cm} \Pr(X^*=1|z)]'$,
\begin{eqnarray}\label{equation: x star conditional on z}
\widehat{\Pr}(X^*|z)=\widehat{M}^{-1}_{X_1|X^*,z}\left[\widehat{\Pr}(X_1=0|z) \hspace{0.3cm} \widehat{\Pr}(X_1=1|z)\right]',
\end{eqnarray}
where $\widehat{\Pr}(X_1=j|z)=\sum\nolimits_{i=1}^{n}  \mathbbm{1}\left(X_{1i}=j, Z_i=z\right)/\sum\nolimits_{i=1}^{n}\mathbbm{1}\left(Z_i=z\right), j=0,1.$

The matrix $M_{X_3|X^*, z}$ can be estimated from equation (\ref{equation: joint distribution of two observations}) and the relationship
$M_{X_3,X^*|z}=M_{X_3|X^*,z}M_{X^*|z}$, where $M_{X^*|z}$ is a diagonal matrix and the two diagonal elements are $\Pr(X^*=0|z)$ and $\Pr(X^*=1|z)$. Thus, we have
\begin{eqnarray}
\widehat{M}_{X_3|X^*,z}=\left(\widehat{M}^{-1}_{X_1|X^*,z}\widehat{M}_{X_1,X_3|z}\right)^{\prime}\hbox{diag}\left[
1/\widehat{\Pr}(X^*=0| z) \hspace{0.3cm} 1/\widehat{\Pr}(X^*=1| z)\right].
\end{eqnarray}
The main advantage of the nonparametric estimation is that it is global and involves no optimization. However, the estimated probability might be outside the interval $[0,1]$ due to finite sample properties. To address this potential issue, we propose an extreme estimator with probabilities constrained to $[0,1]$ based on equation (\ref{equation: main identification equation}).
Let $p_{1k}=\Pr(X_1=1|X^*=k, z)$ and $p_{2k}=\Pr(X_2=x_2|X^*=k, z)$ be the parameters to be estimated for $k=0, 1$. We assume that  Assumption \ref{assumption: invertibility, distinct eigenvalues, and ordering}(ii) holds for $X_1$ such that $p_{10}<p_{11}$. The estimation problem becomes
\begin{eqnarray}\label{equation: nonparametric estimation of eigenvector matrix}
\left(\widehat{p}_{10}, \widehat{p}_{11}, \widehat{p}_{20}, \widehat{p}_{21}\right)&=&\arg\min\left|\left|\widehat{M}_{X_1, x_2,X_3|z}\widehat{M}^{-1}_{X_1, X_3|z}-MDM^{-1}(p_{10}, p_{11}, p_{20}, p_{21},z)\right|\right|\nonumber\\
\hbox{subject to}&& \hspace{-0.3cm} (i)\hspace{0.1cm} (p_{10}, p_{11}, p_{20}, p_{21})\in[0,1]^4\nonumber\\
&&\hspace{-0.3cm} (ii)\hspace{0.1cm}  p_{20}\neq p_{21}\nonumber\\
&&\hspace{-0.3cm} (iii)\hspace{0.1cm} p_{10}<p_{11},
\end{eqnarray}
where the $\left|\left|\cdot\right|\right|$ is a matrix norm (e.g., Frobenius norm) and the constraints are imposed according to Assumption \ref{assumption: invertibility, distinct eigenvalues, and ordering}. Analogously, $\Pr(X^*|z)$ and $\Pr(X_3|X^*,z)$ are  estimated from equation (\ref{equation: estimation of the conditional probabilities of sensitive information}) and (\ref{equation: joint distribution of two observations}), respectively.

\subsection{Continuous convariates}
When the vector of characteristics $Z$ takes a large number of values or is continuous, the nonparametric estimation demands a large sample, which is often unavailable. In this section, we present an MLE method to estimate the latent probabilities.

The MLE is based on the main identification equation (\ref{equation: main identification equation}),
\begin{eqnarray*}
\Pr\left( X_{1}, X_2, X_3|z\right) =\sum\nolimits_{X^*\in\{0,1\}}\Pr\left( X_{1}| X^*,z\right)\Pr\left( X_{2}| X^*,z\right)\Pr\left( X_3|X^*,z\right)\Pr\left(X^*|z\right).
\end{eqnarray*}
We first parametrize the unknown probabilities.
\begin{eqnarray}\label{equation: parametrization of probabilities}
\Pr(X_1=1|X^*=j, z) &\equiv& g(z; \alpha_j), \hspace{0.3cm} \Pr(X_2=1|X^*=j, z) \equiv g(z; \beta_j),\nonumber\\
\Pr(X_3=1|X^*=j, z) &\equiv& g(z; \gamma_j), \hspace{0.3cm} \Pr(X^*=1| z) \equiv g(z; \rho), j=0, 1.
\end{eqnarray}
The $g(\cdot)$ can take various functional forms such as a logistic function. The model parameters $(\alpha_0, \alpha_1, \beta_0, \beta_1, \gamma_0, \gamma_1, \rho)$ are estimated by maximizing the following log-likelihood function.
\begin{eqnarray}
\log \mathcal{L}
&=&\sum\nolimits_{i=1}^n \log
\left\{\sum\nolimits_{X^*\in\{0,1\}}\Pr\left( X_{1i}| X^*,z_i\right)\Pr\left( X_{2i}| X^*,z_i\right)\Pr\left( X_{3i}|X^*,z_i\right)\Pr\left(X^*|z_i\right)\right\}.\nonumber\\
&=&\sum\nolimits_{i=1}^n \log
\bigg\{\sum\nolimits_{j\in\{0,1\}}\big[g(z_i; \alpha_j)\big]^{X_{1i}}\big[1-g(z_i; \alpha_j)\big]^{1-X_{1i}}\big[g(z_i; \beta_j)\big]^{X_{2i}}\big[1-g(z_i; \beta_j)\big]^{1-X_{2i}}\nonumber\\
&&\times
\big[g(z_i; \gamma_j)\big]^{X_{3i}}\big[1-g(z_i; \gamma_j)\big]^{1-X_{3i}}g(z_i;\rho)^j\big[1-g(z_i; \rho)\big]^{1-j}\textcolor{white}{\{\sum_{j\in\{0,1\}}}\hspace{-1.2cm}\bigg\},
\end{eqnarray}
Assumption \ref{assumption: invertibility, distinct eigenvalues, and ordering}(ii) requires $\Pr(X_1=1|X^*=0, z)<\Pr(X_1=1|X^*=1, z)$, which implies that
$
\alpha_0\neq \alpha_1, \beta_0\neq \beta_1, \gamma_0\neq \gamma_1,
$ and $g(z; \alpha_0)<g(z; \alpha_1)$ for all $z$ are needed for identification.

\section{Monte Carlo Experiments\label{section: Monte Carlos simulations}}
In this section, we provide Monte Carlo evidence to illustrate the performance of our proposed methodology.
We consider two settings where the covariate $Z$ is discrete and continuous, respectively. The sample size is set to be 500, 1,000, and 2,000, which is comparable to the sample size in the articles we analyze in section \ref{section: application of the test}. The results are based on 1,000 replications.

In the first setting, $Z$ is a binary variable with the probability of $Z=0$ being 0.4. We set parameters for $\Pr(X^*=1|z)$ and
$\Pr(X_j=1|x^*,z)$ to generate the sample $\{X_{1i}, X_{2i}, X_{3i}, Z_i\}, n=1, 2,\cdots, n.$ The first step of estimation is to test the rank of the $2\times 2$ matrix $M_{X_1, X_3|z}$ by using the sequential testing procedure proposed in \cite{robin2000tests}.
Specifically, we test
\begin{eqnarray*}
 H_0: rank\left(M_{X_1, X_3|z}\right) = 1 \hspace{0.3cm}vs. \hspace{0.3cm} H_1: rank\left(M_{X_1, X_3|z}\right) > 1.
\end{eqnarray*}
We reject the null hypothesis with a 100\% rejection rate for all the three sample sizes.\footnote{The rejection rate is 100\% for a sample size greater than 300.} Next, we estimate the parameters
$\Pr(X^*), \Pr(X^*|z)$, and $\Pr(X_j|x^*, z)$ using the method of  matrix decomposition and the extreme estimator with a Frobenius norm. The estimated results are presented in Table \ref{table: simulation results for discrete z}. As shown in the table, our estimates track the true value closely even for a small sample size $n=500$. As sample size increases, the standard error decreases significantly. The performance of the closed-form estimator and the extreme estimator is similar.

In the second setting, $Z$ is uniformly distributed on $[0,1]$. The data generating process is based on Equation (\ref{equation: parametrization of probabilities}) with $g(\cdot)$ being a logistic  function.
From the model primitives, we generate the joint distribution of $\{X_{1i}, X_{2i}, X_{3i}, Z_i\}$ for $i=1, 2, \cdots, n$. We present the estimation results in Table \ref{table: simulation results for continuous z}. The results show good performance of our estimation method. The estimates improve significantly when sample size increases from $500$ to $2,000$.

A comparison of the results using the nonparametric methods and MLE shows that the nonparametric approach performs better than MLE, especially when the sample size is relatively small, e.g., $n=500$. This is mainly due to the fact that the nonparametric method is global and the estimation involves no optimization. Specifically, a closed-form
estimator is global by construction. By contrast, an optimization algorithm, such as MLE, can only
guarantee a local maximum or minimum even when a global solution exists. The nonparametric approach also allows us to
analyze how parameters affect the estimate constructively while this can only be done numerically
for an estimator using optimization algorithms. In addition, a closed-form estimator is computationally
more convenient since most of the optimization algorithms involve iterations.

In the second set of simulation, we check the sensitivity of our estimates to the assumption of conditional independence, Assumption \ref{assumption: conditional independence}. We allow mutual correlation between $X_j$ and $X_k$ conditional on $X^*$ and $Z$ by assuming the correlation coefficients for any pair of $X$ is a constant $\sigma$.  Then we investigate how our estimate depends on $\sigma$. We present the results of estimation for a discrete $Z$ in Tables \ref{table: simulation results for discrete z, correlation 0.05}-\ref{table: simulation results for discrete z, correlation 0.20} and a continuous $Z$ in table \ref{table: simulation results for continuous z with correlation}. For a model with a discrete $Z$, the estimates for $\sigma=0.05$ and $\sigma=0.10$ are close to $\sigma=0$.   When $\sigma=0.05$, the estimates are very close to the true values, especially when the sample size is 2000. Not surprisingly, when the correlation coefficient increases to $\sigma=0.2$, the bias becomes larger in magnitude. Correlation leads to overestimation. The impact of the correlation is larger for a continuous $Z$. Table \ref{table: simulation results for continuous z with correlation} shows that the estimates are close to the no correlation case only for $\sigma=0.05$.

\begin{table}
\centering
\caption{Simulation results of nonparametric estimation: discrete $Z$}
\vspace{-0.6cm}
\label{table: simulation results for discrete z}
\begin{center}
\scalebox{0.8}{\begin{threeparttable}
\begin{tabular}{lllllllll}
\hline\hline
& &\multicolumn{3}{c}{estimate: closed-form} &&\multicolumn{3}{c}{estimate: extreme estimator} \\
\cline{3-5} \cline{7-9}\vspace{-0.4cm}
\\
parameter&true value &$n=500$ & $n=1000$ & $n=2000$ &  & $n=500$ & $n=1000$ & $n=2000$ \\
\hline
$\Pr(X^*=1)$ & 0.642 & 0.643 & 0.641 & 0.634 &  & 0.641 & 0.642 & 0.634 \\
 &  & (0.062) & (0.042) & (0.029) &  & (0.076) & (0.044) & (0.031) \\
$\Pr(X^*=1|z=0)$ & 0.378 & 0.377 & 0.376 & 0.379 &  & 0.382 & 0.377 & 0.379 \\
 &  & (0.102) & (0.071) & (0.050) &  & (0.103) & (0.071) & (0.052) \\
$\Pr(X^*=1|z=1)$ & 0.818 & 0.808 & 0.813 & 0.817 &  & 0.803 & 0.813 & 0.817 \\
 &  & (0.082) & (0.049) & (0.034) &  & (0.108) & (0.055) & (0.036) \\
 \hline
$\Pr(X_1=1|X^*=0,z=0)$ & 0.269 & 0.265 & 0.270 & 0.267 &  & 0.264 & 0.268 & 0.267 \\
 &  & (0.079) & (0.058) & (0.037) &  & (0.086) & (0.057) & (0.040) \\
$\Pr(X_1=1|X^*=0,z=1)$ & 0.310 & 0.288 & 0.303 & 0.304 &  & 0.297 & 0.305 & 0.304 \\
 &  & (0.191) & (0.112) & (0.073) &  & (0.151) & (0.105) & (0.073) \\
$\Pr(X_1=1|X^*=1,z=0)$ & 0.881 & 0.895 & 0.885 & 0.882 &  & 0.880 & 0.881 & 0.881 \\
 &  & (0.113) & (0.076) & (0.051) &  & (0.078) & (0.062) & (0.048) \\
$\Pr(X_1=1|X^*=1,z=1)$ & 0.900 & 0.904 & 0.903 & 0.900 &  & 0.904 & 0.902 & 0.900 \\
 &  & (0.040) & (0.025) & (0.017) &  & (0.046) & (0.030) & (0.018) \\
 \hline
$\Pr(X_2=1|X^*=0,z=0)$ & 0.269 & 0.266 & 0.270 & 0.268 &  & 0.267 & 0.269 & 0.267 \\
 &  & (0.061) & (0.048) & (0.032) &  & (0.065) & (0.047) & (0.034) \\
$\Pr(X_2=1|X^*=0,z=1)$ & 0.289 & 0.285 & 0.288 & 0.287 &  & 0.285 & 0.288 & 0.287 \\
 &  & (0.120) & (0.080) & (0.055) &  & (0.125) & (0.084) & (0.055) \\
$\Pr(X_2=1|X^*=1,z=0)$ & 0.731 & 0.735 & 0.734 & 0.732 &  & 0.734 & 0.735 & 0.733 \\
 &  & (0.085) & (0.062) & (0.041) &  & (0.083) & (0.060) & (0.040) \\
$\Pr(X_2=1|X^*=1,z=1)$ & 0.750 & 0.752 & 0.752 & 0.750 &  & 0.755 & 0.753 & 0.751 \\
 &  & (0.038) & (0.026) & (0.019) &  & (0.053) & (0.031) & (0.019) \\
 \hline
$\Pr(X_3=1|X^*=0,z=0)$ & 0.269 & 0.262 & 0.266 & 0.267 &  & 0.259 & 0.264 & 0.267 \\
 &  & (0.077) & (0.056) & (0.037) &  & (0.071) & (0.050) & (0.036) \\
$\Pr(X_3=1|X^*=0,z=1)$ & 0.289 & 0.256 & 0.283 & 0.284 &  & 0.285 & 0.285 & 0.283 \\
 &  & (0.212) & (0.107) & (0.076) &  & (0.169) & (0.116) & (0.087) \\
$\Pr(X_3=1|X^*=1,z=0)$ & 0.881 & 0.899 & 0.893 & 0.884 &  & 0.880 & 0.888 & 0.883 \\
 &  & (0.110) & (0.080) & (0.050) &  & (0.100) & (0.069) & (0.052) \\
$\Pr(X_3=1|X^*=1,z=1)$ & 0.891 & 0.896 & 0.893 & 0.892 &  & 0.892 & 0.893 & 0.892 \\
 &  & (0.045) & (0.026) & (0.018) &  & (0.062) & (0.030) & (0.018) \\
 \hline\hline
\end{tabular}
\begin{tablenotes}
\item
\end{tablenotes}
\end{threeparttable}
}
\end{center}
\end{table}

\begin{table}
\centering
\caption{Simulation results: continuous $Z$}
\vspace{-0.6cm}
\label{table: simulation results for continuous z}
\begin{center}
\scalebox{1}{\begin{threeparttable}
\begin{tabular}{lllll}
\hline\hline
& &\multicolumn{3}{c}{estimate} \\
\cline{3-5}\vspace{-0.4cm}
\\
parameter&true value &$n=500$ & $n=1000$ & $n=2000$ \\
\hline
$\rho$ & 1 & 0.964 & 0.989 & 0.985 \\
 &  & (0.396) & (0.273) & (0.193) \\
$\alpha_1$ & 1 & 1.021 & 1.012 & 1.008 \\
 &  & (0.267) & (0.188) & (0.126) \\
$\alpha_0$ & -1 & -1.031 & -1.018 & -0.990 \\
 &  & (0.469) & (0.292) & (0.201) \\
$\beta_1$ & 2 & 2.272 & 2.050 & 2.036 \\
 &  & (1.575) & (0.360) & (0.249) \\
$\beta_0$ & -2 & -2.494 & -2.144 & -2.044 \\
 &  & (2.859) & (0.723) & (0.436) \\
$\gamma_1$ & 2 & 2.286 & 2.050 & 2.030 \\
 &  & (2.235) & (0.378) & (0.240) \\
$\gamma_0$ & -2 & -2.446 & -2.107 & -2.031 \\
 &  & (2.099) & (0.694) & (0.443) \\
 \hline\hline
\end{tabular}
\begin{tablenotes}
\item
\end{tablenotes}
\end{threeparttable}
}
\end{center}
\end{table}

\begin{table}
\centering
\caption{Simulation results of nonparametric estimation: discrete $Z$, $\sigma=0.05$}
\vspace{-0.6cm}
\label{table: simulation results for discrete z, correlation 0.05}
\begin{center}
\scalebox{0.8}{\begin{threeparttable}
\begin{tabular}{lllllllll}
\hline\hline
& &\multicolumn{3}{c}{estimate: closed-form} &&\multicolumn{3}{c}{estimate: extreme estimator} \\
\cline{3-5} \cline{7-9}\vspace{-0.4cm}
\\
parameter&true value &$n=500$ & $n=1000$ & $n=2000$ &  & $n=500$ & $n=1000$ & $n=2000$ \\
\hline
$\Pr(X^*=1|z=0)$ & 0.378 & 0.394 & 0.387 & 0.389 &  & 0.397 & 0.388 & 0.389 \\
 &  & (0.097) & (0.068) & (0.045) &  & (0.093) & (0.067) & (0.045) \\
$\Pr(X*=1|z=1)$ & 0.818 & 0.814 & 0.817 & 0.819 &  & 0.813 & 0.816 & 0.819 \\
 &  & (0.074) & (0.046) & (0.031) &  & (0.073) & (0.046) & (0.031) \\
$\Pr(X^*=1)$ & 0.642 & 0.659 & 0.648 & 0.647 &  & 0.660 & 0.648 & 0.647 \\
 &  & (0.060) & (0.039) & (0.026) &  & (0.059) & (0.039) & (0.026) \\
$\Pr(X_1=1|X^*=0,z=0)$ & 0.269 & 0.246 & 0.254 & 0.254 &  & 0.248 & 0.255 & 0.254 \\
 &  & (0.076) & (0.051) & (0.035) &  & (0.075) & (0.051) & (0.035) \\
$\Pr(X_1=1|X^*=0,z=1)$ & 0.310 & 0.273 & 0.289 & 0.293 &  & 0.279 & 0.289 & 0.293 \\
 &  & (0.165) & (0.101) & (0.070) &  & (0.143) & (0.099) & (0.070) \\
$\Pr(X_1=1|X^*=1,z=0)$ & 0.881 & 0.896 & 0.889 & 0.886 &  & 0.884 & 0.887 & 0.886 \\
 &  & (0.108) & (0.066) & (0.045) &  & (0.082) & (0.061) & (0.044) \\
$\Pr(X_1=1|X^*=1,z=1)$ & 0.900 & 0.903 & 0.903 & 0.901 &  & 0.903 & 0.904 & 0.901 \\
 &  & (0.039) & (0.024) & (0.016) &  & (0.037) & (0.024) & (0.016) \\
$\Pr(X_2=1|X^*=0,z=0)$ & 0.269 & 0.240 & 0.243 & 0.245 &  & 0.242 & 0.244 & 0.245 \\
 &  & (0.062) & (0.043) & (0.030) &  & (0.065) & (0.043) & (0.030) \\
$\Pr(X_2=1|X^*=0,z=1)$ & 0.289 & 0.259 & 0.262 & 0.265 &  & 0.258 & 0.262 & 0.265 \\
 &  & (0.110) & (0.076) & (0.053) &  & (0.110) & (0.076) & (0.053) \\
$\Pr(X_2=1|X^*=1,z=0)$ & 0.731 & 0.719 & 0.721 & 0.719 &  & 0.717 & 0.720 & 0.719 \\
 &  & (0.083) & (0.058) & (0.039) &  & (0.082) & (0.058) & (0.039) \\
$\Pr(X_2=1|X^*=1,z=1)$ & 0.750 & 0.743 & 0.746 & 0.744 &  & 0.744 & 0.746 & 0.744 \\
 &  & (0.039) & (0.026) & (0.019) &  & (0.045) & (0.026) & (0.019) \\
$\Pr(X_3=1|X^*=0,z=0)$ & 0.269 & 0.238 & 0.242 & 0.245 &  & 0.235 & 0.242 & 0.245 \\
 &  & (0.074) & (0.050) & (0.034) &  & (0.070) & (0.049) & (0.034) \\
$\Pr(X_3=1|X^*=0,z=1)$ & 0.289 & 0.240 & 0.259 & 0.261 &  & 0.254 & 0.259 & 0.261 \\
 &  & (0.180) & (0.105) & (0.067) &  & (0.144) & (0.103) & (0.067) \\
$\Pr(X_3=1|X^*=1,z=0)$ & 0.881 & 0.872 & 0.875 & 0.868 &  & 0.866 & 0.873 & 0.868 \\
 &  & (0.102) & (0.068) & (0.047) &  & (0.083) & (0.063) & (0.046) \\
$\Pr(X_3=1|X^*=1,z=1)$ & 0.891 & 0.889 & 0.888 & 0.888 &  & 0.889 & 0.888 & 0.888 \\
 &  & (0.042) & (0.025) & (0.018) &  & (0.037) & (0.025) & (0.018) \\
 \hline\hline
\end{tabular}
\begin{tablenotes}
\item
\end{tablenotes}
\end{threeparttable}
}
\end{center}
\end{table}

\begin{table}
\centering
\caption{Simulation results of nonparametric estimation: discrete $Z$, $\sigma=0.10$}
\vspace{-0.6cm}
\label{table: simulation results for discrete z, correlation 0.10}
\begin{center}
\scalebox{0.8}{\begin{threeparttable}
\begin{tabular}{lllllllll}
\hline\hline
& &\multicolumn{3}{c}{estimate: closed-form} &&\multicolumn{3}{c}{estimate: extreme estimator} \\
\cline{3-5} \cline{7-9}\vspace{-0.4cm}
\\
parameter&true value &$n=500$ & $n=1000$ & $n=2000$ &  & $n=500$ & $n=1000$ & $n=2000$ \\
\hline
$\Pr(X^*=1|z=0)$ & 0.378 & 0.406 & 0.400 & 0.402 &  & 0.406 & 0.401 & 0.402 \\
 &  & (0.089) & (0.060) & (0.042) &  & (0.086) & (0.059) & (0.042) \\
$\Pr(X*=1|z=1)$ & 0.818 & 0.821 & 0.825 & 0.824 &  & 0.819 & 0.825 & 0.824 \\
 &  & (0.057) & (0.041) & (0.028) &  & (0.057) & (0.041) & (0.028) \\
$\Pr(X^*=1)$ & 0.642 & 0.650 & 0.660 & 0.658 &  & 0.819 & 0.825 & 0.824 \\
 &  & (0.050) & (0.034) & (0.024) &  & (0.057) & (0.041) & (0.028) \\
$\Pr(X_1=1|X^*=0,z=0)$ & 0.269 & 0.237 & 0.241 & 0.238 &  & 0.235 & 0.241 & 0.238 \\
 &  & (0.088) & (0.049) & (0.035) &  & (0.073) & (0.049) & (0.035) \\
$\Pr(X_1=1|X^*=0,z=1)$ & 0.310 & 0.261 & 0.266 & 0.276 &  & 0.267 & 0.267 & 0.276 \\
 &  & (0.159) & (0.106) & (0.071) &  & (0.131) & (0.103) & (0.071) \\
$\Pr(X_1=1|X^*=1,z=0)$ & 0.881 & 0.891 & 0.890 & 0.887 &  & 0.888 & 0.888 & 0.887 \\
 &  & (0.097) & (0.062) & (0.040) &  & (0.075) & (0.058) & (0.040) \\
$\Pr(X_1=1|X^*=1,z=1)$ & 0.900 & 0.904 & 0.903 & 0.902 &  & 0.904 & 0.903 & 0.902 \\
 &  & (0.031) & (0.023) & (0.015) &  & (0.033) & (0.023) & (0.015) \\
$\Pr(X_2=1|X^*=0,z=0)$ & 0.269 & 0.217 & 0.217 & 0.217 &  & 0.216 & 0.217 & 0.217 \\
 &  & (0.067) & (0.040) & (0.027) &  & (0.060) & (0.041) & (0.027) \\
$\Pr(X_2=1|X^*=0,z=1)$ & 0.289 & 0.232 & 0.236 & 0.238 &  & 0.232 & 0.236 & 0.238 \\
 &  & (0.102) & (0.072) & (0.050) &  & (0.102) & (0.073) & (0.050) \\
$\Pr(X_2=1|X^*=1,z=0)$ & 0.731 & 0.709 & 0.710 & 0.706 &  & 0.710 & 0.710 & 0.706 \\
 &  & (0.086) & (0.053) & (0.037) &  & (0.080) & (0.053) & (0.037) \\
$\Pr(X_2=1|X^*=1,z=1)$ & 0.750 & 0.736 & 0.734 & 0.735 &  & 0.737 & 0.734 & 0.735 \\
 &  & (0.036) & (0.027) & (0.018) &  & (0.040) & (0.027) & (0.018) \\
$\Pr(X_3=1|X^*=0,z=0)$ & 0.269 & 0.218 & 0.219 & 0.218 &  & 0.214 & 0.219 & 0.218 \\
 &  & (0.079) & (0.047) & (0.031) &  & (0.065) & (0.046) & (0.031) \\
$\Pr(X_3=1|X^*=0,z=1)$ & 0.289 & 0.223 & 0.225 & 0.232 &  & 0.234 & 0.227 & 0.232 \\
 &  & (0.149) & (0.101) & (0.065) &  & (0.129) & (0.098) & (0.065) \\
$\Pr(X_3=1|X^*=1,z=0)$ & 0.881 & 0.857 & 0.856 & 0.851 &  & 0.856 & 0.856 & 0.851 \\
 &  & (0.100) & (0.060) & (0.043) &  & (0.082) & (0.059) & (0.043) \\
$\Pr(X_3=1|X^*=1,z=1)$ & 0.891 & 0.885 & 0.882 & 0.882 &  & 0.885 & 0.882 & 0.882 \\
 &  & (0.035) & (0.025) & (0.017) &  & (0.034) & (0.025) & (0.017) \\
 \hline\hline
\end{tabular}
\begin{tablenotes}
\item
\end{tablenotes}
\end{threeparttable}
}
\end{center}
\end{table}

\begin{table}
\centering
\caption{Simulation results of nonparametric estimation: discrete $Z$, $\sigma=0.20$}
\vspace{-0.6cm}
\label{table: simulation results for discrete z, correlation 0.20}
\begin{center}
\scalebox{0.8}{\begin{threeparttable}
\begin{tabular}{lllllllll}
\hline\hline
& &\multicolumn{3}{c}{estimate: closed-form} &&\multicolumn{3}{c}{estimate: extreme estimator} \\
\cline{3-5} \cline{7-9}\vspace{-0.4cm}
\\
parameter&true value &$n=500$ & $n=1000$ & $n=2000$ &  & $n=500$ & $n=1000$ & $n=2000$ \\
\hline
$\Pr(X^*=1|z=0)$ & 0.378 & 0.428 & 0.427 & 0.428 &  & 0.428 & 0.427 & 0.428 \\
 &  & (0.070) & (0.049) & (0.034) &  & (0.069) & (0.049) & (0.034) \\
$\Pr(X*=1|z=1)$ & 0.818 & 0.829 & 0.832 & 0.831 &  & 0.827 & 0.831 & 0.831 \\
 &  & (0.050) & (0.034) & (0.024) &  & (0.052) & (0.035) & (0.024) \\
$\Pr(X^*=1)$ & 0.642 & 0.667 & 0.672 & 0.667 &  & 0.666 & 0.672 & 0.667 \\
 &  & (0.042) & (0.029) & (0.019) &  & (0.042) & (0.029) & (0.019) \\
$\Pr(X_1=1|X^*=0,z=0)$ & 0.269 & 0.206 & 0.208 & 0.209 &  & 0.206 & 0.208 & 0.209 \\
 &  & (0.063) & (0.044) & (0.031) &  & (0.063) & (0.044) & (0.031) \\
$\Pr(X_1=1|X^*=0,z=1)$ & 0.310 & 0.234 & 0.238 & 0.247 &  & 0.243 & 0.240 & 0.247 \\
 &  & (0.151) & (0.102) & (0.061) &  & (0.123) & (0.093) & (0.061) \\
$\Pr(X_1=1|X^*=1,z=0)$ & 0.881 & 0.895 & 0.891 & 0.888 &  & 0.892 & 0.890 & 0.888 \\
 &  & (0.071) & (0.048) & (0.033) &  & (0.065) & (0.048) & (0.033) \\
$\Pr(X_1=1|X^*=1,z=1)$ & 0.900 & 0.903 & 0.903 & 0.903 &  & 0.904 & 0.903 & 0.903 \\
 &  & (0.030) & (0.020) & (0.014) &  & (0.033) & (0.021) & (0.014) \\
$\Pr(X_2=1|X^*=0,z=0)$ & 0.269 & 0.155 & 0.156 & 0.155 &  & 0.155 & 0.156 & 0.155 \\
 &  & (0.050) & (0.034) & (0.025) &  & (0.050) & (0.034) & (0.025) \\
$\Pr(X_2=1|X^*=0,z=1)$ & 0.289 & 0.177 & 0.185 & 0.183 &  & 0.178 & 0.184 & 0.183 \\
 &  & (0.090) & (0.063) & (0.044) &  & (0.089) & (0.063) & (0.044) \\
$\Pr(X_2=1|X^*=1,z=0)$ & 0.731 & 0.680 & 0.678 & 0.679 &  & 0.680 & 0.678 & 0.679 \\
 &  & (0.069) & (0.050) & (0.034) &  & (0.069) & (0.050) & (0.034) \\
$\Pr(X_2=1|X^*=1,z=1)$ & 0.750 & 0.718 & 0.717 & 0.717 &  & 0.718 & 0.717 & 0.717 \\
 &  & (0.035) & (0.024) & (0.017) &  & (0.038) & (0.026) & (0.017) \\
$\Pr(X_3=1|X^*=0,z=0)$ & 0.269 & 0.157 & 0.158 & 0.159 &  & 0.156 & 0.158 & 0.159 \\
 &  & (0.059) & (0.039) & (0.028) &  & (0.057) & (0.039) & (0.028) \\
$\Pr(X_3=1|X^*=0,z=1)$ & 0.289 & 0.154 & 0.165 & 0.172 &  & 0.175 & 0.169 & 0.172 \\
 &  & (0.203) & (0.093) & (0.065) &  & (0.119) & (0.086) & (0.065) \\
$\Pr(X_3=1|X^*=1,z=0)$ & 0.881 & 0.826 & 0.824 & 0.822 &  & 0.825 & 0.824 & 0.822 \\
 &  & (0.079) & (0.050) & (0.037) &  & (0.076) & (0.050) & (0.037) \\
$\Pr(X_3=1|X^*=1,z=1)$ & 0.891 & 0.875 & 0.873 & 0.874 &  & 0.875 & 0.873 & 0.874 \\
 &  & (0.034) & (0.023) & (0.016) &  & (0.033) & (0.023) & (0.016) \\
 \hline\hline
\end{tabular}
\begin{tablenotes}
\item
\end{tablenotes}
\end{threeparttable}
}
\end{center}
\end{table}

\begin{table}
\centering
\caption{Simulation results: continuous $Z$ with correlation}
\vspace{-0.6cm}
\label{table: simulation results for continuous z with correlation}
\begin{center}
\scalebox{0.8}{\begin{threeparttable}
\begin{tabular}{lllllllllllll}
\hline\hline
& &\multicolumn{3}{c}{$\sigma=0.05$}&&\multicolumn{3}{c}{$\sigma=0.10$}&&\multicolumn{3}{c}{$\sigma=0.20$} \\
\cline{3-5}\cline{7-9}\cline{11-13}\vspace{-0.4cm}
\\
parameter&\tabincell{c}{true\\ value} &$500$ & $1000$ & $2000$&&$500$ & $1000$ & $2000$&&$500$ & $1000$ & $2000$\\
\hline
$\rho$ & 1 & 1.015 & 1.017 & 1.028 &  & 1.053 & 1.065 & 1.071 &  & 1.078 & 1.104 & 1.106 \\
 &  & (0.365) & (0.268) & (0.187) &  & (0.354) & (0.238) & (0.180) &  & (0.309) & (0.202) & (0.141) \\
$\alpha_1$ & 1 & 1.044 & 1.037 & 1.034 &  & 1.074 & 1.059 & 1.066 &  & 1.143 & 1.118 & 1.125 \\
 &  & (0.259) & (0.183) & (0.126) &  & (0.243) & (0.173) & (0.124) &  & (0.231) & (0.170) & (0.116) \\
$\alpha_0$ & -1 & -1.106 & -1.086 & -1.083 &  & -1.215 & -1.213 & -1.195 &  & -1.462 & -1.417 & -1.400 \\
 &  & (0.424) & (0.303) & (0.211) &  & (0.448) & (0.301) & (0.202) &  & (0.469) & (0.304) & (0.211) \\
$\beta_1$ & 2 & 2.068 & 1.997 & 1.954 &  & 1.976 & 1.896 & 1.882 &  & 1.835 & 1.761 & 1.757 \\
 &  & (0.889) & (0.342) & (0.212) &  & (0.671) & (0.289) & (0.208) &  & (0.465) & (0.269) & (0.176) \\
$\beta_0$ & -2 & -2.740 & -2.446 & -2.335 &  & -3.184 & -2.794 & -2.731 &  & -4.484 & -3.895 & -3.698 \\
 &  & (2.109) & (0.914) & (0.495) &  & (2.645) & (0.846) & (0.556) &  & (6.401) & (1.651) & (0.755) \\
$\gamma_1$ & 2 & 2.071 & 1.999 & 1.959 &  & 2.003 & 1.889 & 1.868 &  & 1.849 & 1.766 & 1.748 \\
 &  & (0.780) & (0.372) & (0.216) &  & (0.811) & (0.288) & (0.193) &  & (0.804) & (0.245) & (0.169) \\
$\gamma_0$ & -2 & -2.725 & -2.398 & -2.340 &  & -3.187 & -2.816 & -2.703 &  & -4.972 & -3.896 & -3.687 \\
 &  & (2.019) & (0.822) & (0.501) &  & (2.633) & (0.886) & (0.590) &  & (9.221) & (1.254) & (0.781) \\
 \hline\hline
\end{tabular}
\begin{tablenotes}
\item
\end{tablenotes}
\end{threeparttable}
}
\end{center}
\end{table}

\end{document}